[1994/12/01]
\documentclass[manuscript, a4paper, reqno]{aomart}

\chardef\bslash=`\\

\newtheorem[{}\it]{thm}{Theorem}[section]
\newtheorem{cor}[thm]{Corollary}
\newtheorem{lem}[thm]{Lemma}

\newtheorem{op}[thm]{Open Problem}
\newtheorem{prop}[thm]{Proposition}

\theoremstyle{remark} 

\newtheorem*{dang}{\textsc{Theorem}}

\usepackage{dsfont}
\usepackage{amssymb}
\usepackage[misc]{ifsym}

\usepackage{graphicx}
\usepackage{moresize}

\usepackage{hyperref}
\usepackage[numeric]{amsrefs}

\usepackage[english]{babel}

\theoremstyle{definition}
\newtheorem{defn}{\textsc{Definition}}[section]
\newtheorem{rem}{Remark}[section]

\newtheorem*[{}\it]{notation}{Notation}

\newtheorem*[{}\it]{rest}{\textsc{Theorem}}
\newtheorem*[{}\it]{quest}{\textsc{Question}}
\newtheorem*[{}\it]{problemo}{\textsc{Problem}}
\newtheorem*[{}\it]{projone}{\texttt{Project 1 (Future Work)}}
\newtheorem*[{}\it]{projtwo}{\texttt{Project 2 (Current Work)}}
\newtheorem*[{}\it]{projthree}{\texttt{Project 3 (Current Work)}}
\newtheorem*[{}\it]{projfour}{\texttt{Project 4 (Future Work)}}
\newtheorem*[{}\it]{projfive}{\texttt{Project 5 (Future Work)}}
\newtheorem*[{}\it]{proofoflemma}{Proof of Lemma}

\usepackage[textwidth=6in,textheight=9.5in, paperwidth=7.5in,paperheight=11.25in, inner=2cm, outer=2cm]{geometry} 

\usepackage[usenames,dvipsnames]{xcolor}

\usepackage{txfonts}

\usepackage{multicol}

\usepackage{enumerate}

\title[]{On the Convergence of Soft Potential Dynamics to Hard Sphere Dynamics}
\author[]{Mark Wilkinson}\thanks{ Courant Institute of Mathematical Sciences, New York University, New York City. (\Letter) \href{mailto:mwilkins@cims.nyu.edu}{mwilkins@cims.nyu.edu}}

\allowdisplaybreaks

\usepackage{tikz}
\usetikzlibrary{shapes.geometric, arrows}

\definecolor{ashgrey}{rgb}{0.7, 0.75, 0.71}

\tikzstyle{startstop} = [rectangle, rounded corners, minimum width=3cm, minimum height=0.7cm,text centered, draw=black, fill=w]

\tikzstyle{blank} = [rectangle, rounded corners, minimum width=3cm, minimum height=1cm,text centered, draw=white, fill=white]

\tikzstyle{io} = [trapezium, trapezium left angle=70, trapezium right angle=110, minimum width=3cm, minimum height=0.7cm, text centered, draw=black, fill=white]

\tikzstyle{iog} = [trapezium, trapezium left angle=70, trapezium right angle=110, minimum width=3cm, minimum height=0.7cm, text centered, draw=black, fill=ashgrey]

\tikzstyle{process} = [rectangle, minimum width=3cm, minimum height=0.7cm, text centered, draw=black, fill=white]

\tikzstyle{processmall} = [rectangle, minimum width=1.2cm, minimum height=0.8cm, text centered, draw=black, fill=white]

\tikzstyle{decision} = [diamond, minimum width=3cm, minimum height=1cm, text centered, draw=black, fill=white]

\tikzstyle{arrow} = [thick,->,>=stealth]

\newcommand{\bv}{\mathrm{BV}_{\mathrm{loc}}(\mathbb{R})}

\newcommand{\Ds}{\mathcal{D}_{2}^{\mathrm{S}}}

\newcommand{\ov}{\overline}

\usepackage{tikz}
\usetikzlibrary{patterns}
\newcommand{\boundellipse}[3]
{(#1) ellipse (#2 and #3)
}

\begin{document}

\maketitle

\begin{abstract}
\noindent We address a question raised in the work of \textsc{Gallagher, Saint-Raymond and Texier} \cite{MR3157048} that concerns the convergence of soft-potential dynamics to hard sphere dynamics. In the case of two particles, we establish that hard sphere dynamics is the limit of soft sphere dynamics in the weak-star topology of  $\mathrm{BV}$. We view our result as establishing a {\em topological} method by which to construct weak solutions to the ODE of hard sphere motion. 
\end{abstract}
%

\vspace{3mm}

\section{Introduction}
In this article, we consider topological methods by which one can establish the existence of weak solutions to the equations of `physical' hard sphere motion. As such, our starting point shall be a system of $N$ identical spherical particles in $\mathbb{R}^{3}$ (without loss of generality, each of unit diameter and of mass 1) whose motion is governed by the Hamiltonian $H^{\varepsilon}_{N}$ given by
\begin{equation}\label{hammy}
H^{\varepsilon}_{N}(X, V):=\frac{1}{2}\sum_{k=1}^{N}|v_{k}|^{2}+\sum_{i\neq j}\Phi^{\varepsilon}(x_{i}-x_{j}),
\end{equation}
where $X=[x_{1}, ..., x_{N}]$, $V=[v_{1}, ..., v_{N}]$, $x_{i}, v_{i}\in\mathbb{R}^{3}$, $0<\varepsilon<1$ and the potentials $\Phi^{\varepsilon}:\mathbb{R}^{3}\rightarrow [0, \infty]$ have the property that (i) they are compactly supported on $\mathbb{R}$, (ii) are spherically symmetric and smooth on $\mathbb{R}^{3}\setminus\{0\}$, (iii) are radially decreasing on $\mathbb{R}^{3}\setminus\{0\}$ and (iv) blow up at the origin. This Hamiltonian is consistent with Newton's Laws of Motion, in that it is has both translation and rotation symmetry in phase space, and is also time-independent, which formally imply the conservation of linear momentum, angular momentum and kinetic energy for its associated dynamics, respectively (see \textsc{Arnol'd} \cite{MR997295}). We consider the asymptotic behaviour of the system of Hamiltonian ODEs associated to \eqref{hammy} as the potential $\Phi^{\varepsilon}$ is made to {\em harden}, namely $\Phi^{\varepsilon}\rightarrow \Phi$ in an appropriate topology as $\varepsilon\rightarrow 0$, where
\begin{equation}\label{hardpotential}
\Phi(y):=\left\{
\begin{array}{ll}
0 & \quad \text{if} \hspace{2mm} |y|\leq 1, \vspace{2mm}\\
\infty & \quad \text{otherwise}.
\end{array}
\right.
\end{equation}
In the monograph of \textsc{Gallagher, Saint-Raymond and Texier} (\cite{MR3157048}, p.2) on the validity of the Boltzmann-Grad limit for systems of soft or hard spheres, the authors remark that ``the dynamics of hard spheres is in some sense the limit of the smooth-forces case''. Indeed, in this article we prove that for $N=2$, hard sphere dynamics is the limit of soft-potential dynamics as $\varepsilon\rightarrow 0$ in the weak-$\ast$ topology on $\mathrm{BV}(I, \mathbb{R}^{6})$ for any open interval $I\subset\mathbb{R}$. An informal statement of our main result is as follows:
\begin{dang}\label{mainthminf} {\em
Let $\{\Phi^{\varepsilon}\}_{0<\varepsilon<1}$ be a suitable family of soft potentials that converges to $\Phi$ as $\varepsilon\rightarrow 0$. Suppose initial conditions $Z_{0}:=[x_{0}, \ov{x}_{0}, v_{0}, \ov{v}_{0}]\in\mathbb{R}^{12}$ for two spheres each of unit diameter are taken such that $|x-\ov{x}|\geq 1$. If $v^{\varepsilon}=v^{\varepsilon}(\cdot; Z_{0})$ and $\ov{v}^{\varepsilon}(\cdot; Z_{0})$ denote solutions to the equations of motion associated with $H^{\varepsilon}_{2}$, then one has
\begin{equation*}
\|v^{\varepsilon}\|_{\mathrm{BV}(I, \mathbb{R}^{3})}+\|\ov{v}^{\varepsilon}\|_{\mathrm{BV}(I, \mathbb{R}^{3})}\leq C(Z_{0}, I)
\end{equation*}  
for any open interval $I=(a, b)\subset \mathbb{R}$ and some constant $C=C(Z_{0}, I)>0$ independent of the hardening parameter $\varepsilon$. Moreover, $[v^{\varepsilon}, \ov{v}^{\varepsilon}]$ converges in $L^{1}_{\mathrm{loc}}(\mathbb{R}, \mathbb{R}^{6})$ to the unique classical solution $[v, \ov{v}]$ of the equations of hard sphere motion associated to the singular Hamiltonian $H_{2}:\mathbb{R}^{12}\rightarrow [0, \infty]$, where 
\begin{equation*}
H_{2}(x, \ov{x}, v, \ov{v}):=\frac{1}{2}(|v|^{2}+|\ov{v}|^{2})+\Phi(x-\ov{x}).
\end{equation*}
}
\end{dang} 
The refined statement of this result (with precise hypotheses on the potentials $\{\Phi^{\varepsilon}\}_{0<\varepsilon<1}$) appears in \ref{mainres} below. This theorem establishes, in a precise sense, that the qualitative properties of soft sphere systems are close to those of hard sphere systems when $0<\varepsilon\ll 1$, i.e. when $\Phi^{\varepsilon}$ is `sufficiently close' to $\Phi$. However, one can view the softening of the potential $\Phi$ via $\Phi^{\varepsilon}$ as a topological method by which to construct {\em weak solutions} to the ODEs of physical hard sphere dynamics associated with $\Phi$. 

Due to the more complicated estimates arising from simultaneous $M$-particle collisions ($2<M\leq N$), we do not consider the case of systems of $N\geq 3$ spheres in this article. In the final section of the paper, we discuss the problem of construction of physical dynamics for two hard {\em non-spherical} particles.
\subsection{Some Results in the Literature}
At the heart of this paper, we are interested in the existence and regularity of solutions to the equations of motion for $N$ hard spheres in $\mathbb{R}^{3}$. Mathematically, this amounts to the construction of a dynamics on (a suitable subset of) the high-dimensional phase space 
\begin{equation*}
\mathcal{D}_{N}:=\left\{Z_{N}=[(x_{1}, v_{1}), ..., (x_{N}, v_{N})]\in\mathbb{R}^{6N}\,:\, |x_{i}-x_{j}|\geq 1\hspace{2mm}\text{for}\hspace{2mm}i\neq j\right\}
\end{equation*}
which is also subject to constraints on velocity (namely the linear momentum, angular momentum and kinetic energy of the system must be constant in time). It is well known that one can define a global-in-time $N$-particle trajectory on $\mathcal{D}_{N}$ for `most' initial data $Z_{0}\in\mathcal{D}_{N}$. More precisely, one has the following statement:
\begin{prop}[\cite{MR3157048}, proposition 4.1.1]
Let $\mu_{N}$ denote the restriction of the $6N$-dimensional Lebesgue measure to the phase space $\mathcal{D}_{N}$. The set of `bad' initial data $\mathcal{B}_{N}\subset\mathcal{D}_{N}$ which give rise to either (i) grazing collisions, (ii) simultaneous collisions involving $M\geq 3$ spheres, or (iii) infinitely-many collisions in a finite time interval is of $\mu_{N}$-measure zero.
\end{prop}
With this observation one can construct, by means of the method of trajectory surgery, global-in-time classical solutions to the equations of motion for a set of full $\mu_{N}$-measure in $\mathcal{D}_{N}$; see section \ref{mts} for details on this method of construction and section \ref{notion} below for the definition of classical solution in the case $N=2$. As we have an existence theory for a `large' subset of initial data in $\mathcal{D}_{N}$, one can in turn ask about qualitative properties of $N$-particle trajectories starting from data therein. In particular, one might wish to know the maximum number of collisions associated to an initial datum $Z_{0}\in\mathcal{D}_{N}\setminus\mathcal{B}_{N}$. Indeed, this is a difficult problem: see, for instance, the review article of \textsc{Murphy and Cohen} (\cite{MR1805337}, chapter 1).

The fact that one only has an existence theory for the equations of motion on a full-measure set is fine, of course, if one is only concerned with the study of statistical dynamics on $\mathcal{D}_{N}$ (for instance, the Boltzmann-Grad limit for $N$ hard spheres on $\mathbb{R}^{3}$). However, the lack of an existence theory for all initial data $Z_{0}\in\mathcal{B}_{N}$ may be unsatisfying to the analyst. To the knowledge of the author, there is no existence and regularity theory for either `classical' or `weak' solutions to the equations of motion for initial data $Z_{0}\in\mathcal{B}_{N}$. In particular, it seems no analogue of {\em scattering map} for $M$-particle collisions (for $M\geq 3$) has been constructed and studied, i.e. a map that resolves the collision between 3 or more hard spheres by mapping `pre-collisional' velocities to `post-collisional' velocities in such a way that total linear momentum, angular momentum and kinetic energy is conserved. Mathematically, for a given collision configuration graph $G$ of $M$ hard spheres in $\mathbb{R}^{3}$ and `pre-collisional' initial velocities $v_{1}, ..., v_{M}\in\mathbb{R}^{3}$, one must find `post-collisional' velocities $v_{1}'(G), ..., v_{M}'(G)$ which satisfy the conservation of total linear momentum 
\begin{equation*}
\sum_{i=1}^{M}v_{i}'(G)=\sum_{i=1}^{M}v_{i},
\end{equation*}
the conservation of angular momentum (with respect to any point of measurement $a\in\mathbb{R}^{3}$)
\begin{equation*}
\sum_{i=1}^{M}(x_{i}-a)\wedge v_{i}'(G)=\sum_{i=1}^{M}(x_{i}-a)\wedge v_{i},
\end{equation*}
and the conservation of kinetic energy
\begin{equation*}
\sum_{i=1}^{M}|v_{i}'(G)|^{2}=\sum_{i=1}^{M}|v_{i}|^{2}.
\end{equation*}
We illustrate this problem schematically in figure 1 above.
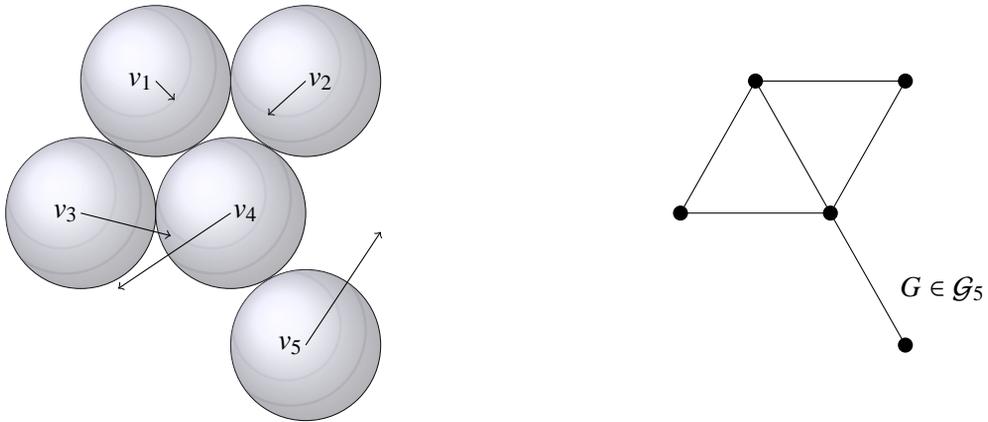
\begin{figure}\label{5ball}
\begin{tikzpicture}
  
    \draw (0,0) circle (1cm);
    \shade[ball color=blue!10!white,opacity=0.20] (0,0) circle (1cm);
    
     \draw (2,0) circle (1cm);
    \shade[ball color=blue!10!white,opacity=0.20] (2,0) circle (1cm);
    
     \draw (1,1.75) circle (1cm);
    \shade[ball color=blue!10!white,opacity=0.20] (1,1.75) circle (1cm);
    
 \draw (3,1.75) circle (1cm);
    \shade[ball color=blue!10!white,opacity=0.20] (3,1.75) circle (1cm);
     \draw (3,-1.75) circle (1cm);
    \shade[ball color=blue!10!white,opacity=0.20] (3,-1.75) circle (1cm);
    
   \fill[black] (8,0) circle (0.1cm);
     \fill[black] (10,0) circle (0.1cm);
       \fill[black] (11,1.75) circle (0.1cm);
         \fill[black] (9,1.75) circle (0.1cm);
           \fill[black] (11,-1.75) circle (0.1cm);
           
           \draw (8,0) -- (10,0);
            \draw (8,0) -- (9,1.75);
             \draw (9,1.75) -- (11,1.75);
              \draw (9,1.75) -- (10,0);
              \draw (11, 1.75) -- (10,0);
               \draw (10,0) -- (11,-1.75);
               
        \node(draw) at (11.5,-1) {$G\in\mathcal{G}_{5}$};
        
        \draw [->] (0, 0) -- (1.2, -0.3);
        \draw [->] (2, 0) -- (0.5, -1);
           \draw [->] (3, -1.75) -- (4, -0.25);
              \draw [->] (1, 1.75) -- (1.25, 1.5);
                 \draw [->] (3, 1.75) -- (2.5, 1.3);
                 
                  \node(draw) at (0.8,1.75) {$v_{1}$};
                  \node(draw) at (3.2,1.75) {$v_{2}$};
                       \node(draw) at (-0.2,0) {$v_{3}$};
                       \node(draw) at (2.2,0) {$v_{4}$};
                       \node(draw) at (2.8,-1.75) {$v_{5}$};
\end{tikzpicture}
\caption{A planar configuration of 5 hard spheres in $\mathbb{R}^{3}$ in simultaneous collision. The configuration is characterised by a graph $G$ with five nodes and 6 edges, each of length 1. For given `pre-collisional' velocities $[v_{1}, ..., v_{5}]\in\mathbb{R}^{15}$, one would like to construct `post-collisional' velocities $[v_{1}'(G), ..., v_{5}'(G)]\in\mathbb{R}^{15}$ which conserve total linear momentum, angular momentum and kinetic energy of the initial datum.}
\end{figure}
The corresponding scattering map $\sigma_{G}:\mathbb{R}^{3M}\rightarrow\mathbb{R}^{3M}$ (whose domain is not, in general, all of $\mathbb{R}^{3M}$) is given simply by $\sigma_{G}[v_{1}, ..., v_{M}]:=[v_{1}'(G), ..., v_{M}'(G)]$. It is natural to stipulate also, for instance, that $\sigma_{G}$ is an involution on $\mathbb{R}^{3M}$ and that $\mathrm{det}(D\sigma_{G}[V])=-1$ for all $V$ in the domain of $\sigma_{G}$. In any case, if one could construct families of scattering maps $\{\sigma_{G}\}_{G\in\mathcal{G}_{M}}$ corresponding to $M$-particle collisions (with $3\leq M\leq N$ and $\mathcal{G}_{M}$ being the class of all graphs parametrising $M$-particle collisions), then the general existence theory of \textsc{Ballard} \cite{MR1785473} allows one to establish the global-in-time existence of weak solutions to the equations of $N$-particle motion for arbitrary initial data in (a suitable subset of) $\mathcal{D}_{N}$. This theory can also be applied to the problem of non-spherical particle motion, but is only immediately applicable to the case when the boundary manifolds of the particles are real analytic. We shall say more about this in the final section of the article.

\subsection{`Algebraic' and `Topological' Constructions of Weak Solutions}
Let us denote the hard sphere of unit diameter whose centre of mass lies at $y\in\mathbb{R}^{3}$ by $\mathsf{S}(y)$. The equations of motion for two hard spheres are given {\em formally} by
\begin{equation}\label{formal}
\frac{d}{dt}\left[
\begin{array}{c}
x \\
\ov{x}
\end{array}
\right]=\left[
\begin{array}{c}
v \\
\ov{v}
\end{array}
\right]\quad \text{and} \quad \frac{d}{dt}\left[
\begin{array}{c}
v \\
\ov{v}
\end{array}
\right]=\left[
\begin{array}{c}
0 \\
0
\end{array}
\right],
\end{equation}
where the centres of mass $x$ and $\ov{x}$ are constrained to satisfy the condition $|x(t)-\ov{x}(t)|\geq 1$ for all $t\in\mathbb{R}$. Suppose that two hard spheres $\mathsf{S}(x(t))$ and $\mathsf{S}(\ov{x}(t))$ are in collision with one another at a {\em collision time} $t=\tau$, namely
\begin{equation}\label{collconf}
\mathrm{card}\,\mathsf{S}(x(\tau))\cap\mathsf{S}(\ov{x}(\tau))=1\quad \text{with}\quad \ov{x}(\tau)=x(\tau)+n,
\end{equation}
for some unit vector $n\in\mathbb{S}^{2}$. The problem of understanding how to resolve a collision between $\mathsf{S}(x(\tau))$ and $\mathsf{S}(\ov{x}(\tau))$ in such a way that (i) there is conservation of total linear momentum, angular momentum\footnote{We draw attention to the fact that angular momentum is rarely considered for the problem of two colliding spheres. However, it is shown in section \ref{scatty} that conservation of angular momentum allows us to solve for the `post-collisional' velocities in a systematic manner.} and kinetic energy of the two spheres, and (ii) they do not overlap following collision, has been well understood since the work of \textsc{Boltzmann} \cite{boltzmann2012wissenschaftliche}. Indeed, following the construction of a velocity scattering matrix $\sigma_{n}$ for two hard spheres (which is essentially an algebraic problem), one performs what we term in this article `trajectory surgery' to join pre-collisional 2-particle trajectories to post-collisional ones that yield classical solutions of \eqref{formal}. As perhaps indicated by the statement of the above theorem \ref{mainthminf}, we focus our attention in this article on the topological method of construction of weak solutions of \eqref{formal} in $\mathrm{BV}_{\mathrm{loc}}(\mathbb{R})$, a natural functional space in which to obtain compactness of families of smooth approximate trajectories $Z^{\varepsilon}$. Let us now briefly review the well-known construction of classical solutions to system \eqref{formal} by the method of trajectory surgery, before discussing our new contribution to this problem.
\subsection{`Algebraic' Construction of Classical and Weak Solutions: The Method of Trajectory Surgery} \label{algebraic}
We begin by noting that the set of all admissible phase points for the evolution of two hard spheres is the set of positions and velocities
\begin{equation*}
\mathcal{D}_{2}(\mathsf{S}_{\ast}):=\left\{Z=[z, \ov{z}]\in\mathbb{R}^{12}\,:\,\mathrm{card}(\mathsf{S}_{\ast}+x)\cap(\mathsf{S}_{\ast}+\ov{x})\leq 1\right\},
\end{equation*}
where $\mathsf{S}_{\ast}\subset\mathbb{R}^{3}$ is the sphere of unit diameter and centre at the origin, and $z=[x, v]$, $z=[\ov{x}, \ov{v}]$ denote the phase points of each individual hard sphere. Of particular interest is the boundary of this set, 
\begin{equation*}
\partial\mathcal{D}_{2}(\mathsf{S}_{\ast})=\left\{Z\in\mathbb{R}^{12}\,:\,|x-\ov{x}|=1\right\},
\end{equation*}
which constitutes the set of all {\em collision configurations} of two hard spheres in $\mathbb{R}^{3}$. 

The form of the ODE system \eqref{formal} clearly suggests that particle trajectories are rectilinear in the interior of the phase space $\mathcal{D}_{2}(\mathsf{S}_{\ast})$, i.e. when initial conditions $Z_{0}=[x_{0}, \ov{x}_{0}, v_{0}, \ov{v}_{0}]\in\mathcal{D}_{2}(\mathsf{S}_{\ast})$ are taken such that $|x_{0}-\ov{x}_{0}|>1$, then $x(t):=x_{0}+tv_{0}$ and $\ov{x}(t):=\ov{x}_{0}+t\ov{v}_{0}$ solve the system \eqref{formal} pointwise in the classical sense on some (possibly short) time interval. However, when the two hard spheres come into collision with one another (otherwise said, when $Z(\tau)\in\mathcal{D}_{2}(\mathsf{S}_{\ast})$ for some $\tau\in\mathbb{R}$), we must find a way of updating the particle velocities so that $Z(t)\in\mathcal{D}_{2}(\mathsf{S}_{\ast})$ for $t>\tau$. Aside from this spatial constraint, one also stipulates the velocity constraint that the collision conserves total linear momentum, angular momentum and kinetic energy of the particle system. To do this, one must construct a family of scattering matrices $\{\sigma_{n}\}_{n\in\mathbb{S}^{2}}$ which map `pre-collisional' velocities to `post-collisional' ones.

\subsubsection{Construction of Physical Scattering}\label{scatty}
The collision resolution is typically formulated as a family of algebraic problems (parametrised by the vector $n$ in \eqref{collconf}) for the unknown post-collisional linear velocities $v_{n}', \ov{v}_{n}'$ of $\mathsf{S}(x(\tau)), \mathsf{S}(\ov{x}(\tau))$, respectively. Let us consider this in detail. Suppose the spheres in collision possess `pre-collisional'\footnote{The reason we encase the word {\em pre-collisional} in inverted commas is that we have not yet specified in precise terms which $v, \ov{v}\in\mathbb{R}^{3}$ constitute pre-collisional velocity vectors with respect to the spatial configuration $n\in\mathbb{S}^{2}$. This is an issue related to regularity of the dynamics $t\mapsto [x(t), \ov{x}(t)]$, and is of greater significance when we consider systems of $M>2$ hard spheres, or systems of non-spherical particles.} linear velocities $v, \ov{v}\in\mathbb{R}^{3}$. One looks to find `post-collisional' linear velocities $v_{n}', \ov{v}_{n}'$ such that the conservation of total linear momentum
\begin{equation}\label{colm}
v_{n}'+\ov{v}_{n}'=v+\ov{v},\tag{COLM}
\end{equation}
the conservation of angular momentum (with respect to an arbitrary point of measurement $a\in\mathbb{R}^{3}$)
\begin{equation}\label{coam}
-(a-x(\tau))\wedge v_{n}'-(a-x(\tau)-n)\wedge \ov{v}_{n}'=-(a-x(\tau))\wedge v-(a-x(\tau)-n)\wedge \ov{v},\tag{COAM}
\end{equation}
and the conservation of kinetic energy
\begin{equation}\label{coke}
|v_{n}'|^{2}+|\ov{v}_{n}'|^{2}=|v|^{2}+|\ov{v}|^{2}.\tag{COKE}
\end{equation}
hold true. Although \eqref{coam} ought to hold for arbitrary points of measurement $a\in\mathbb{R}^{3}$, to simplify the problem we choose it to be the centre of mass of the system, namely $a=\frac{1}{2}n$. We may also suppose, by using \eqref{colm} directly, that $x(\tau)=0$. These choices generate 6 linear equations and one quadratic equation in the 6 unknowns $v_{n}', \ov{v}_{n}'$. Recasting \eqref{colm} and \eqref{coam} as the linear system $E_{n}V_{n}'=E_{n}V$, where
\begin{equation*}
E_{n}:=\left(
\begin{array}{cccccc}
1 & 0 & 0 & 1 & 0 & 0 \\
0 & 1 & 0 & 0 & 1 & 0 \\
0 & 0 & 1 & 0 & 0 & 1 \\
0 & -n_{3} & n_{2} & 0 & n_{3} & -n_{2}\\
n_{3} & 0 & -n_{1} & -n_{3} & 0 & n_{1} \\
-n_{2} & n_{1} & 0 & n_{2} & -n_{1} & 0
\end{array}
\right)
\end{equation*}
and $V_{n}':=[v_{n}', \ov{v}_{n}']$, $V:=[v, \ov{v}]$, it may be quickly checked that $E_{n}$ is singular for every choice of $n\in\mathbb{S}^{2}$, i.e. \eqref{colm} and \eqref{coam} give rise to at most 5 independent linear equations. Now, setting the first component of $v_{n}'$ to be $v_{n, 1}'=\alpha$ for some parameter $\alpha\in\mathbb{R}$, using the two linear conservation laws one may express all other unknown components of $v_{n}'$ and $\ov{v}_{n}'$ in terms of $\alpha$ alone. In turn, substitution of $v_{n}'$ and $\ov{v}_{n}'$ into \eqref{coke} yields a quadratic equation in the single unknown $\alpha$. One solution of this quadratic yields the trivial solution $v_{n}'=v$ and $\ov{v}_{n}'=\ov{v}$ for every $n\in\mathbb{S}^{2}$. The only other solution is the so-called classical {\em Boltzmann scattering} given by
\begin{equation}\label{boltzmannscattering}
\left[
\begin{array}{c}
v_{n}' \\
\ov{v}_{n}'
\end{array}
\right]=(\underbrace{I-2\widehat{\nu}_{n}\otimes \widehat{\nu}_{n}}_{\sigma_{n}:=})\left[
\begin{array}{c}
v \\
\ov{v}
\end{array}
\right],
\end{equation}
where $\widehat{\nu}_{n}\in\mathbb{S}^{5}$ is the unit vector
\begin{equation*}
\widehat{\nu}_{n}:=\frac{1}{\sqrt{2}}\left[
\begin{array}{c}
n \\
-n
\end{array}
\right].
\end{equation*}
Note that $\sigma_{n}=I-2\widehat{\nu}_{n}\otimes\widehat{\nu}_{n}\in\mathrm{O}(6)$ is a reflection matrix which maps the `lower' half-space $\Sigma_{n}^{-}:=\{V\in\mathbb{R}^{6}\,:\,V\cdot\widehat{\nu}_{n}\geq 0\}$ to the `upper' half-space $\Sigma_{n}^{+}:=\{V\in\mathbb{R}^{6}\,:\,V\cdot\widehat{\nu}_{n}\leq 0\}$ for each $n\in\mathbb{S}^{2}$. Moreover, we note that once the trivial solution has been discarded, Newton's law of restitution for perfectly-elastic impacts, namely
\begin{equation}\label{spherenon}
(v_{n}'-\ov{v}_{n}')\cdot n=-(v-\ov{v})\cdot n,
\end{equation}
is a simple consequence of the posited conservation laws \eqref{colm}, \eqref{coam} and \eqref{coke}. With the family of matrices $\{\sigma_{n}\}_{n\in\mathbb{S}^{2}}$ in hand, one can now construct global-in-time trajectories by the method of trajectory surgery, which we now present.

\subsubsection{The Method of Trajectory Surgery}\label{mts}
The following algorithm allows one to construct a map $Z:\mathbb{R}\rightarrow\mathcal{D}_{2}(\mathsf{S}_{\ast})$ which ensures non-penetration of two hard spheres, and which also respects the fundamental conservation laws of classical mechanics.
\begin{enumerate}[(I)]
\item START: Select an initial datum $Z_{0}\in\mathcal{D}_{2}(\mathsf{S}_{\ast})$. Consider the associated globally-defined linear trajectory $t\mapsto Z_{1}(t)=[z_{1}(t), \ov{z}_{1}(t)]$ in $\mathbb{R}^{12}$, where
\begin{equation*}
\left[
\begin{array}{c}
x_{1}(t) \\
\ov{x}_{1}(t)
\end{array}
\right]:=\left[
\begin{array}{c}
x_{0}+tv_{0}\\
\ov{x}_{0}+t\ov{v}_{0}
\end{array}
\right]\quad \text{and} \quad \left[
\begin{array}{c}
v_{1}(t) \\
v_{2}(t) \\
\end{array}
\right]:=\left[
\begin{array}{c}
v_{0}\\
\ov{v}_{0}
\end{array}
\right].
\end{equation*}
\item Define the set of collision times $\mathcal{T}(Z_{0}):=\left\{
t\in\mathbb{R}\,:\,Z_{1}(t)\in\partial\mathcal{D}_{2}(\mathsf{S}_{\ast})
\right\}.$
\begin{enumerate}[i.]
\item If $\mathcal{T}(Z_{0})=\varnothing$, then set the solution $Z(t):=Z_{1}(t)$ for all $t\in\mathbb{R}$, and STOP; otherwise GO TO (II--ii.).
\item If $\mathcal{T}(Z_{0})=\mathbb{R}$, then set the solution $Z(t):=Z_{1}(t)$ for all $t\in\mathbb{R}$, and STOP; otherwise GO TO (III).
\end{enumerate}
\item Define $\tau:=\min\mathcal{T}(Z_{0})$. 
\begin{enumerate}[i.]
\item If there exists $\delta>0$ such that $|x_{1}(t)-\ov{x}_{1}(t)|>1$ for $\tau-\delta<t<\tau$, `perform surgery' on $Z_{1}$ using the scattering matrix $\sigma_{n}$ to define $Z_{2}:=[x_{2}, \ov{x}_{2}, v_{2}, \ov{x}_{2}]$ with $X_{2}=[x_{2}, \ov{x}_{2}]$ and $V=[v_{2}, \ov{v}_{2}]$ given by
\begin{equation*}
X_{2}(t)=\left\{
\begin{array}{ll}
X_{0}+tV_{0} & \quad \text{for} \hspace{2mm}t\leq \tau, \vspace{2mm}\\
X_{0}+\tau V_{0}+(t-\tau)\sigma_{n}V_{0} & \quad \text{for}\hspace{2mm}t>\tau,
\end{array}
\right.
\end{equation*}
and
\begin{equation}\label{velmap}
V_{2}(t)=\left\{
\begin{array}{ll}
V_{0} & \quad \text{for}\hspace{2mm} t\leq \tau, \vspace{2mm} \\
\sigma_{n}V_{0} & \quad \text{for}\hspace{2mm}t>\tau,
\end{array}
\right.
\end{equation}
where $n:=\ov{x}_{0}-x_{0}$, and STOP; otherwise GO TO (III--ii.).
\item If there exists $\delta>0$ such that $|x_{1}(t)-\ov{x}_{1}(t)|<1$ for $\tau-\delta<t<\tau$, `perform surgery' on $Z_{1}$ using the scattering matrix $\sigma_{n}^{-1}=\sigma_{n}$ to define $Z_{2}:=[x_{2}, \ov{x}_{2}, v_{2}, \ov{x}_{2}]$ with $X_{2}=[x_{2}, \ov{x}_{2}]$ and $V=[v_{2}, \ov{v}_{2}]$ given by
\begin{equation*}
X_{2}(t)=\left\{
\begin{array}{ll}
X_{0}+\tau V_{0}+(t-\tau)\sigma_{n}V_{0} & \quad \text{for}\hspace{2mm}t\leq \tau, \vspace{2mm}\\
X_{0}+tV_{0} & \quad \text{for} \hspace{2mm}t>\tau,
\end{array}
\right.
\end{equation*}
and
\begin{equation}\label{velmap}
V_{2}(t)=\left\{
\begin{array}{ll}
\sigma_{n}V_{0} & \quad \text{for}\hspace{2mm}t>\tau, \vspace{2mm}\\
V_{0} & \quad \text{for}\hspace{2mm} t> \tau,

\end{array}
\right.
\end{equation}
where $n:=\ov{x}_{0}-x_{0}$, and STOP; otherwise GO TO (III--iii.).
\item If $|x_{1}(t)-\ov{x}_{1}(t)|>0$ for {\em both} $t<\tau$ and $t>\tau$, set $Z_{2}(t):=Z_{1}(t)$ and STOP.
\end{enumerate}
\end{enumerate}
By employing the above algorithm, one constructs $Z:\mathbb{R}\rightarrow\mathcal{D}_{2}(\mathsf{S}_{\ast})$ with the property that $x, \ov{x}$ are continuous and both left- and right-differentiable everywhere on $\mathbb{R}$, while $v, \ov{v}$ are lower semi-continuous and left-differentiable everywhere on $\mathbb{R}$. Moreover, $Z=Z(t)$ satisfies the system of one-sided ODEs
\begin{equation*}
\frac{d}{dt_{-}}\left[
\begin{array}{c}
x \\
\ov{x} \\
v \\
\ov{v}
\end{array}
\right]=\left[
\begin{array}{c}
v_{-} \\
\ov{v}_{-}\\
0\\
0
\end{array}
\right]\quad \text{pointwise everywhere on}\hspace{2mm}\mathbb{R},
\end{equation*}
and
\begin{equation*}
\frac{d}{dt_{+}}\left[
\begin{array}{c}
x \\
\ov{x} \\
v \\
\ov{v}
\end{array}
\right]=\left[
\begin{array}{c}
v_{+} \\
\ov{v}_{+} \\
0\\
0
\end{array}
\right]\quad \text{pointwise everywhere on}\hspace{2mm}\mathbb{R}\setminus\mathcal{T}(Z_{0}).
\end{equation*}
Moreover, using the identity \eqref{spherenon}, it follows that 
\begin{equation*}
|x(t)-\ov{x}(t)|\geq 1\quad\text{for all}\hspace{2mm}t\in\mathbb{R},
\end{equation*}
while $Z=Z(t)$ conserves the total linear momentum, angular momentum and kinetic energy of its initial datum $Z_{0}$. In the language of section \ref{notion} below, we have constructed a global-in-time classical solution to system \eqref{formal}. In particular, since the Boltzmann scattering matrix \eqref{boltzmannscattering} is the unique matrix respecting the conservation of total linear momentum, angular momentum and kinetic energy, it follows that this classical solution is unique.

The method of trajectory surgery is particularly straightforward in the case of only two spherical particles. It becomes more complicated in the case of $M\geq 2$ spherical particles, and still more complicated when spheres are replaced by compact, strictly-convex sets whose boundary surfaces are of class $C^{1}$. Let us now set up our topological method for construction of solutions to system \eqref{formal} above.
\section{Preliminaries and Notation}
In all that follows, we consider behaviour of the 2-body Hamiltonians
\begin{equation*}
H_{2}^{\varepsilon}(x, \ov{x}, v, \ov{v})=\frac{1}{2}\left(|v|^{2}+|\ov{v}|^{2}\right)+\Phi^{\varepsilon}(x-\ov{x}).
\end{equation*}
in the limit $\varepsilon\rightarrow 0$ such that $\Phi^{\varepsilon}\rightarrow\Phi$ in a suitable topology. We now set out the properties we require of the soft potentials $\Phi^{\varepsilon}$ in this article.
\subsection{Hypotheses on the Potentials $\Phi^{\varepsilon}$}
The family of spherically-symmetric 2-body potentials $\{\Phi^{\varepsilon}\}_{0<\varepsilon<1}$ is defined in terms of a {\em reference potential} $\Phi_{0}:(0, \infty)\rightarrow (0, \infty)$ by
\begin{equation*}
\Phi^{\varepsilon}(x):=\frac{\Phi_{0}(|x|)}{\varepsilon}.
\end{equation*}
We suppose that $\Phi_{0}$ belongs to the class of all maps on $(0, \infty)$ satisfying the following properties:
\begin{enumerate}[(P1)]
\item\label{p1} $\Phi_{0}$ is of class $C^{2}((0, \infty))$, $\mathrm{supp}(\Phi_{0})=\{r\,:\,r\leq 1\}$, $\Phi_{0}'(r)<0$ for all $r\in(0, 1)$ and $\Phi_{0}''(r)>0$ for all $r\in (0,1)$. Moreover, 
\begin{equation*}
\lim_{r\rightarrow 0+}\Phi_{0}(r)=\infty;
\end{equation*}
\item\label{p2} There exist constants $0<c_{1}<c_{2}$, $\beta>2$ and $0<r_{0}<1$ such that
\begin{equation*}
c_{1}(1-r)^{\beta}\leq\Phi_{0}(r)\leq c_{2}(1-r)^{\beta}
\end{equation*}
for all $r_{0}\leq r\leq 1$;
\item\label{p3} There exist constants $0<\kappa_{1}<\kappa_{2}$ such that
\begin{equation*}
\kappa_{1}(1-r)^{\beta-1}\leq\left|\frac{d\Phi_{0}}{dr}(r)\right|\leq \kappa_{2}(1-r)^{\beta-1}
\end{equation*}
for all $r_{0}\leq r\leq 1$, where $r_{0}$ is as in (P2).
\end{enumerate}
One such family of potentials $\{\Phi^{\varepsilon}\}_{0<\varepsilon<1}$ is generated by the radial functions $\Phi_{0}$ defined by
\begin{equation*}
\Phi_{0}(r):=\left\{
\begin{array}{ll}
r^{-s}(1-r)^{\beta} & \quad \text{for}\hspace{2mm}0<r\leq 1, \vspace{2mm}\\
0 & \quad \text{otherwise},
\end{array}
\right.
\end{equation*}
where $s<0$ and $\beta>2$.
\begin{rem}
For the purposes of approximating hard sphere dynamics, one need not work only with reference potentials $\Phi_{0}$ which have (and whose first derivatives have) algebraic decay at the boundary of their support. As such, properties (P2) and (P3) could be made more general. However, in this work, such a family of potentials is sufficient to establish a compactness result in $\bv$.
\end{rem}
\subsection{Notation}
For brevity, we shall often use the shorthand $Z\in\mathbb{R}^{12}$ to denote the phase vector $[x, \ov{x}, v, \ov{v}]$ which characterises the state of a system of two hard or soft spheres. Accordingly, the soft sphere phase space $\mathcal{D}_{2}^{\mathrm{S}}$ for two bodies is given by
\begin{equation*}
\mathcal{D}_{2}^{\mathrm{S}}:=\left\{Z=[x, \ov{x}, v, \ov{v}]\in\mathbb{R}^{12}\,:\,x\neq\ov{x}\right\}.
\end{equation*}
while the hard sphere phase space $\mathcal{D}_{2}(\mathsf{S}_{\ast})$ is given by
\begin{equation*}
\mathcal{D}_{2}(\mathsf{S}_{\ast}):=\left\{Z=[x, \ov{x}, v, \ov{v}]\in\mathbb{R}^{12}\,:\,|x-\ov{x}|\geq 1\right\},
\end{equation*}
If $\mathcal{D}$ denotes either  $\mathcal{D}_{2}^{\mathrm{S}}$ or $\mathcal{D}_{2}(\mathsf{S}_{\ast})$, we denote by $\Pi_{1}:\mathcal{D}\rightarrow\mathbb{R}^{6}$ the spatial projection operator $\Pi_{1}Z:=[x, \ov{x}]$ and by $\Pi_{2}:\mathcal{D}\rightarrow\mathbb{R}^{6}$ the velocity projection operator $\Pi_{2}Z:=[v, \ov{v}]$. We shall often denote $\Pi_{1}Z$ and $\Pi_{2}Z$ simply by $X$ and $V$, respectively; furthermore, $(\Pi_{1}Z)_{1}:=x$, $(\Pi_{1}Z)_{2}:=\ov{x}$, $(\Pi_{2}Z)_{1}:=v$ and $(\Pi_{2}Z)_{2}:=\ov{v}$. We define the linear momentum functional $\mathrm{LM}:\mathcal{D}\rightarrow\mathbb{R}^{3}$ of a phase point $Z\in\mathcal{D}$ by
\begin{equation*}
\mathrm{LM}(Z):=mv+m\ov{v},
\end{equation*}
the angular momentum functional $\mathrm{AM}:\mathcal{D}\times\mathbb{R}^{3}\rightarrow\mathbb{R}^{3}$ (with respect to a point of measurement $a\in\mathbb{R}^{3}$) by
\begin{equation*}
\mathrm{AM}(Z; a):=-m(a-x)\wedge v-m(a-\ov{x})\wedge\ov{v},
\end{equation*}
and the kinetic energy functional $\mathrm{KE}:\mathcal{D}\rightarrow [0, \infty)$ by
\begin{equation*}
\mathrm{KE}(Z):=m|v|^{2}+m|\ov{v}|^{2}.
\end{equation*}
We write $C_{0}^{k}(\mathbb{R}, \mathbb{R}^{M})$ (often simply denoted by $C^{k}_{0}(\mathbb{R})$) to denote the space of $k$-times differentiable maps with compact support in $\mathbb{R}$ equipped with the norm
\begin{equation*}
\|\phi\|_{C^{k}_{0}(\mathbb{R})}:=\sum_{i=0}^{k} \max_{t\in\mathbb{R}}|\phi^{(k)}(t)|.
\end{equation*}
Finally, if $a=(a_{1}, a_{2})\in\mathbb{R}^{2}$, we denote by $a^{\perp}$ the orthogonal vector $(-a_{2}, a_{1})$.
\subsection{Notions of Solution to the Equations of Motion}\label{notion}
While the dynamics associated with $H^{\varepsilon}_{2}$ is smooth, hard sphere trajectories $t\mapsto Z(t)$ are inherently non-smooth due to the non-penetration constraint, i.e. that dynamics $t\mapsto Z(t)$ must have range in $\mathcal{D}_{2}(\mathsf{S}_{\ast})$. As such, we must specify the precise senses in which the equations of motion associated with both $H_{2}^{\varepsilon}$ and $H_{2}$ can be satisfied. Firstly, for each $0<\varepsilon<1$, the equations of motion for soft spheres read as
\begin{equation}\label{softodes}
\frac{d}{dt}\left[
\begin{array}{c}
x^{\varepsilon} \\
v^{\varepsilon} 
\end{array}
\right]=\left[
\begin{array}{c}
v^{\varepsilon} \\
-\nabla\Phi^{\varepsilon}(x^{\varepsilon}-\ov{x}^{\varepsilon})
\end{array}
\right] \quad \text{and} \quad \frac{d}{dt}\left[
\begin{array}{c}
\ov{x}^{\varepsilon} \\
\ov{v}^{\varepsilon} 
\end{array}
\right]=\left[
\begin{array}{c}
\ov{v}^{\varepsilon} \\
\nabla\Phi^{\varepsilon}(x^{\varepsilon}-\ov{x}^{\varepsilon})
\end{array}
\right]. \tag{S$^{\varepsilon}$}
\end{equation}
We subsequently work with only one notion of solution to system \eqref{softodes}.
\begin{defn}[Classical Solutions of \eqref{softodes}]
For a given initial datum $Z_{0}\in\mathcal{D}_{2}^{\mathrm{S}}$, a {\em classical solution} of system (S$^{\varepsilon}$) is a map $Z^{\varepsilon}=[x^{\varepsilon}, \ov{x}^{\varepsilon}, v^{\varepsilon}, \ov{v}^{\varepsilon}]\in C^{1}(\mathbb{R}, \mathbb{R}^{12})$ whose components satisfy the equations \eqref{softodes} pointwise on $\mathbb{R}$ for all time and $Z^{\varepsilon}(0)=Z_{0}$. Moreover, $Z^{\varepsilon}$ satisfies the conservation of linear momentum
\begin{equation}\label{thmcolm}
\mathrm{LM}(Z^{\varepsilon}(t))=\mathrm{LM}(Z_{0}),
\end{equation}
the conservation of angular momentum (with respect to any point of measurement $a\in\mathbb{R}^{3}$)
\begin{equation}\label{thmcoam}
\mathrm{AM}(Z^{\varepsilon}(t); a)=\mathrm{AM}(Z_{0}; a),
\end{equation}
and the conservation of kinetic energy
\begin{equation}\label{thmcoke}
\mathrm{KE}(Z^{\varepsilon}(t))=\mathrm{KE}(Z_{0}),
\end{equation}
for all time $t\in\mathbb{R}$.
\end{defn}
In contrast, the equations of motion for hard sphere dynamics are
\begin{equation}\label{sminus}
\frac{d}{dt_{-}}\left[
\begin{array}{c}
x \\
\ov{x}
\end{array}
\right]=\left[
\begin{array}{c}
v_{-} \\
\ov{v}_{-}
\end{array}
\right] \quad \text{and} \quad \frac{d}{dt_{-}}\left[
\begin{array}{c}
v \\
\ov{v}
\end{array}
\right]=\left[
\begin{array}{c}
0 \\
0
\end{array}
\right],
\tag{S\textsuperscript{--}}
\end{equation}
and also
\begin{equation}\label{splus}
\frac{d}{dt_{+}}\left[
\begin{array}{c}
x \\
\ov{x}
\end{array}
\right]=\left[
\begin{array}{c}
v_{+} \\
\ov{v}_{+}
\end{array}
\right] \quad \text{and} \quad \frac{d}{dt_{+}}\left[
\begin{array}{c}
v \\
\ov{v}
\end{array}
\right]=\left[
\begin{array}{c}
0 \\
0
\end{array}
\right].
\tag{S\textsuperscript{+}}
\end{equation}
As we have observed above, the ODEs have been separated into their left- and right-limits due to the general non-differentiable corners in the loci $t\mapsto x(t)$ and $t\mapsto \ov{x}(t)$ at collision. We shall deal with two notions of solution to the equations of hard sphere motion in this article.
\begin{defn}[Weak Solutions of \eqref{sminus} and \eqref{splus}]
For a given initial datum $Z_{0}\in\mathcal{D}_{2}(\mathsf{S}_{\ast})$, we say that $Z$ is a {\bf weak solution} of \eqref{sminus} and \eqref{splus} if and only if $\Pi_{1}Z\in C(\mathbb{R}, \mathbb{R}^{6})$ and $\Pi_{2}Z\in\mathrm{BV}_{\mathrm{loc}}(\mathbb{R}, \mathbb{R}^{6})$ satisfy the equations
\begin{equation*}
\int_{-\infty}^{\infty}\left[
\begin{array}{c}
x(t) \\
\ov{x}(t)
\end{array}
\right]\cdot\phi'(t)\,dt=-\int_{\infty}^{\infty}\left[
\begin{array}{c}
v(t) \\
\ov{v}(t)
\end{array}
\right]\cdot\phi(t)\,dt
\end{equation*}
for all $\phi\in C^{1}_{0}(\mathbb{R}, \mathbb{R}^{6})$, and
\begin{equation*}
\int_{-\infty}^{\infty}\left[
\begin{array}{c}
v(t) \\
\ov{v}(t)
\end{array}
\right]\cdot\psi'(t)\,dt=-\int_{-\infty}^{\infty}\psi \,dDV
\end{equation*}
for all $\psi\in C^{1}_{0}(\mathbb{R}, \mathbb{R}^{6})$, where $DV$ denotes a finite vector-valued Radon measure on $\mathbb{R}$. Furthermore, $Z$ respects the conservation of linear momentum \eqref{thmcolm},
the conservation of angular momentum (with respect to any point of measurement $a\in\mathbb{R}^{3}$) \eqref{thmcoam},
and the conservation of kinetic energy \eqref{thmcoke} for any representative of the equivalence class $Z$ and almost every time $t\in\mathbb{R}$.
\end{defn}
For the purposes of defining classical solutions, we make the following definition.
\begin{defn}
For any $Z_{0}\in\mathcal{D}_{2}(\mathsf{S}_{\ast})$, we define the set of all {\bf collision times} $\mathcal{T}(Z_{0})$ for a trajectory $Z:\mathbb{R}\rightarrow\mathcal{D}_{2}(\mathsf{S}_{\ast})$ (satisfying $Z(0)=Z_{0}$) by
\begin{equation*}
\mathcal{T}(Z_{0}):=\left\{t\in\mathbb{R}\,:\, |x(t)-\ov{x}(t)|=1\right\}.
\end{equation*}
\end{defn}
We contrast the notion of weak solution with the following notion of classical solution.
\begin{defn}[Classical Solutions of \eqref{sminus} and \eqref{splus}]
We say that $Z:\mathbb{R}\rightarrow\mathcal{D}_{2}(\mathsf{S}_{\ast})$ is a {\bf classical solution} of \eqref{sminus} and \eqref{splus} if and only if $t\mapsto \Pi_{1}Z$ is continuous piecewise linear and left- and right-differentiable on $\mathbb{R}$, and $t\mapsto\Pi_{2}Z$ is lower semi-continuous piecewise constant\footnote{We adopt the convention that a vector-valued map is lower semi-continuous if and only if its component maps are themselves lower semi-continuous.}, with $\Pi_{1}Z$ and $\Pi_{2}Z$ satisfying (S\textsuperscript{--}) and (S\textsuperscript{+}) on $\mathbb{R}$ and $\mathbb{R}\setminus\mathcal{T}(Z_{0})$, respectively. Moreover, $t\mapsto Z(t)$ must satisfy \eqref{thmcolm}, \eqref{thmcoam} and \eqref{thmcoke} for every $t\in\mathbb{R}$ and all points of measurement $a\in\mathbb{R}^{3}$. 
\end{defn}
We note that every classical solution of system \eqref{sminus} and \eqref{splus} generates a weak solution thereof.
\subsection{Main Result}\label{mainres}
A precise statement of the main result in this article is the following:
\begin{thm}\label{mainresy}
Suppose the reference potential $\Phi_{0}$ satsifies {\em (P1)}, {\em (P2)} and {\em (P3)}. For any $Z_{0}\in\mathcal{D}_{2}(\mathsf{S}_{\ast})$, let $\{Z^{\varepsilon}\}_{0<\varepsilon<1}$ denote the associated unique classical solution of \eqref{softodes}. There exist $v, \ov{v}\in \mathrm{BV}_{\mathrm{loc}}(\mathbb{R}, \mathbb{R}^{3})$ such that $v^{\varepsilon}\overset{\ast}{\rightharpoonup}v$ and $\ov{v}^{\varepsilon}\overset{\ast}{\rightharpoonup}\ov{v}$ in $\mathrm{BV}(I, \mathbb{R}^{3})$ for any open interval $I\subset\mathbb{R}$ as $\varepsilon\rightarrow 0$, where $[x, \ov{x}, v, \ov{v}]$ is a weak solution of \eqref{sminus} and \eqref{splus}. Moreover, the equivalence class $[x, \ov{x}, v, \ov{v}]$ is represented by the unique classical solution of \eqref{sminus} and \eqref{splus} corresponding to the initial datum $Z_{0}$.
\end{thm}
\subsection{Structure of the Article}
In section \ref{softy}, we study basic properties of solutions of system \eqref{softodes}, in particular obtaining explicit estimates on the total time of collision of soft spheres that depend on the hardening parameter $\varepsilon$. In section \ref{bvcomp}, we prove that families of solutions $\{\Pi_{2}Z^{\varepsilon}\}_{0<\varepsilon<1}$ of \eqref{softodes} are pre-compact in the weak-$\ast$ topology on $\mathrm{BV}(I, \mathbb{R}^{6})$ for any open interval $I\subset\mathbb{R}$. In section \ref{close}, we conclude the proof of the main theorem \ref{mainresy}. In section \ref{close}, we close by considering the challenges posed by the analogous problem for two {\em non-spherical} particles.
\section{Properties of Solutions of the Soft Sphere System (S$^{\varepsilon}$)}\label{softy}
By the Cauchy-Lipschitz theorem, for each $0<\varepsilon<1$ and each initial datum $Z_{0}\in\mathcal{D}_{2}^{\mathrm{S}}$, system \eqref{softodes} has a unique global-in-time $C^{1}$ solution $Z^{\varepsilon}:\mathbb{R}\rightarrow\mathcal{D}_{2}^{\mathrm{S}}$ such that $Z^{\varepsilon}(0)=Z_{0}$. As such, for each fixed $\varepsilon$ we have a well-defined family of solution operators $\{T^{\varepsilon}_{t}\}_{t\in\mathbb{R}}$ on $\mathcal{D}_{2}^{\mathrm{S}}$. We shall be interested in obtaining some precise information on the qualitative behaviour of solutions to \eqref{softodes}, which will be of use when establishing a compactness principle in $\bv$ for sequences of `approximate trajectories' $\{Z^{\varepsilon}\}_{\varepsilon>0}$ to hard sphere trajectories.

The motion of the centres of mass of soft spheres is rectilinear when their supports do not intersect, i.e. when $|x^{\varepsilon}(t)-\ov{x}^{\varepsilon}(t)|>1$. When the distance between their centres of mass is strictly less than 1, we can expect their motion to be curvilinear (and, in particular, symmetric with respect to the {\em apse line}: see \eqref{apseline} below). We shall obtain precise information on these curvilinear trajectories, notably the duration of time $\Delta\tau^{\varepsilon}(Z_{0})$ for which the supports of the soft spheres intersect. In all the sequel, we shall refer to the event when the centres of mass of the two soft spheres are such that $|x^{\varepsilon}(t)-\ov{x}^{\varepsilon}(t)|<1$ a {\em collision}.
\subsection{Pre- and Post-collisional Configurations}
Our main focus in what follows will be the study of {\em scattering operators} $\sigma^{\varepsilon}$ which map pre-collisional configurations of soft spheres to post-collisional ones. To do this, we require the following ancillary definition. 
\begin{defn}[Entrance and Exit Times]
For any initial datum $Z_{0}\in\mathcal{D}_{2}^{\mathrm{S}}$ and its associated solution $Z^{\varepsilon}$ of \eqref{softodes}, we write $\tau_{-}^{\varepsilon}=\tau_{-}^{\varepsilon}(Z_{0})$ and $\tau_{+}^{\varepsilon}=\tau_{+}^{\varepsilon}(Z_{0})$ to denote the entrance and exit times for the supports of the soft spheres, respectively, where
\begin{equation*}
\begin{array}{c}
\tau_{-}^{\varepsilon}(Z_{0}):=\inf\left\{t\in\mathbb{R}\,:\,\mathrm{card}\,\mathsf{S}(x^{\varepsilon}(t))\cap\mathsf{S}(\ov{x}^{\varepsilon}(t))=1\right\}, \vspace{2mm} \\
\tau_{+}^{\varepsilon}(Z_{0}):=\sup\left\{t\in\mathbb{R}\,:\,\mathrm{card}\,\mathsf{S}(x^{\varepsilon}(t))\cap\mathsf{S}(\ov{x}^{\varepsilon}(t))=1\right\}.
\end{array}
\end{equation*}
If $Z_{0}\in\mathcal{D}_{2}^{\mathrm{S}}$ is chosen such that no soft sphere collision takes place, namely $|x^{\varepsilon}(t)-\ov{x}^{\varepsilon}(t)|>1$ for all $t\in\mathbb{R}$, we write $\tau_{-}^{\varepsilon}(Z_{0})=-\infty$ and $\tau_{+}^{\varepsilon}(Z_{0})=\infty$. 
\end{defn}
In the case when $\tau_{-}^{\varepsilon}=-\infty$, we note that the unique solution $Z^{\varepsilon}$ of \eqref{softodes} exhibits rectilinear motion for all $t\in\mathbb{R}$; in the case when $\tau_{-}^{\varepsilon}>-\infty$, the associated unique soft sphere trajectories exhibit rectilinear motion for $t\leq \tau_{-}^{\varepsilon}$ and $t\geq \tau_{+}^{\varepsilon}$. Finally, whenever $\tau_{-}^{\varepsilon}>-\infty$, we denote the {\em duration of collision} $\tau_{+}^{\varepsilon}-\tau_{-}^{\varepsilon}$ by $\Delta\tau^{\varepsilon}(Z_{0})$.

We now wish to understand which initial data $Z_{0}$ are {\em pre}-collisional, and which are {\em post}-collisional. To do this, we now consider the auxiliary function $F:\mathbb{R}^{6}\rightarrow\mathbb{R}$ defined by
\begin{equation*}
F(x, \ov{x}):=|x-\ov{x}|^{2}-1.
\end{equation*} 
Evidently, the supports of the soft spheres intersect at a single point if and only if $F(x^{\varepsilon}(\tau), \ov{x}^{\varepsilon}(\tau))=0$ for some $\tau\in\mathbb{R}$. By a simple calculation, one has that at the entrance time $t=\tau_{-}^{\varepsilon}$
\begin{equation*}
\frac{d}{dt}F(x^{\varepsilon}(t), \ov{x}^{\varepsilon}(t))\bigg|_{t=\tau_{-}^{\varepsilon}}\leq 0,
\end{equation*}
whence 
\begin{equation*}
(x^{\varepsilon}(\tau_{-}^{\varepsilon})-\ov{x}^{\varepsilon}(\tau_{-}^{\varepsilon}))\cdot (v^{\varepsilon}(\tau_{-}^{\varepsilon})-\ov{v}^{\varepsilon}(\tau_{-}^{\varepsilon}))\leq 0.
\end{equation*}
In a similar way, one also has that at the exit time $t=\tau_{+}^{\varepsilon}$
\begin{equation*}
\frac{d}{dt}F(x^{\varepsilon}(t), \ov{x}^{\varepsilon}(t))\bigg|_{t=\tau_{+}^{\varepsilon}}\geq 0,
\end{equation*}
which yields
\begin{equation*}
(x^{\varepsilon}(\tau_{+}^{\varepsilon})-\ov{x}^{\varepsilon}(\tau_{+}^{\varepsilon}))\cdot (v^{\varepsilon}(\tau_{+}^{\varepsilon})-\ov{v}^{\varepsilon}(\tau_{+}^{\varepsilon}))\geq 0.
\end{equation*}
These simple observations motivate the following definitions.
\begin{defn}[Pre- and Post-collisional Configurations]
The set of all {\bf pre-collisional} configurations $\Sigma^{-}\subset\Ds$ is given by
\begin{equation*}
\Sigma^{-}:=\left\{Z\in\mathbb{R}^{12}\,:\,|x-\ov{x}|=1\hspace{2mm}\text{and}\hspace{2mm}(x-\ov{x})\cdot (v-\ov{v})\leq 0\right\}
\end{equation*}
and the set of all {\bf post-collisional} configurations $\Sigma^{+}\subset\Ds$ is given by
\begin{equation*}
\Sigma^{+}:=\left\{Z\in\mathbb{R}^{12}\,:\,|x-\ov{x}|= 1\hspace{2mm}\text{and}\hspace{2mm}(x-\ov{x})\cdot (v-\ov{v})\geq 0\right\}.
\end{equation*}
The set of all {\bf grazing collisions} $\Sigma^{0}\subset\Ds$ is given by
\begin{equation*}
\Sigma^{0}:=\left\{Z\in\mathbb{R}^{12}\,:\,|x-\ov{x}|=1 \hspace{2mm}\text{and}\hspace{2mm}(x-\ov{x})\cdot (v-\ov{v})=0\right\}.
\end{equation*}
In particular,  $\Sigma^{-}\cup\Sigma^{+}\cup\Sigma^{0}=\partial\mathcal{D}_{2}(\mathsf{S}_{\ast})$.
\end{defn}
As we are only interested in how the velocity maps $v^{\varepsilon}$ and $\ov{v}^{\varepsilon}$ of the soft spheres are modified following a collision, in the remainder of this article we shall always assume $Z_{0}\in\Sigma^{-}$, whence $\tau_{-}^{\varepsilon}=0$ for all $\varepsilon>0$.

We are now in a position to define the main object of study in this section.
\begin{defn}[Soft Sphere Scattering Maps]\label{defscat}
If $\{T^{\varepsilon}_{t}\}_{t\in\mathbb{R}}$ denotes the 1-parameter family of solution operators on $\mathcal{D}_{2}^{\mathrm{S}}$ associated with \eqref{softodes}, we define the {\em scattering map} $\sigma^{\varepsilon}:\Sigma^{-}\rightarrow\Sigma^{+}$ to be
\begin{equation}\label{scatteringmap}
\sigma^{\varepsilon}Z_{0}:=T^{\varepsilon}_{\tau_{+}^{\varepsilon}(Z_{0})}Z_{0}.
\end{equation}
\end{defn}
Our study of $\sigma^{\varepsilon}$ will involve two elements:
\begin{itemize}
\item By locating the so-called apse line, we find an explicit formula for the operator $\sigma^{\varepsilon}$; 
\item We study the behaviour of the scattering operator $\sigma^{\varepsilon}$ on $\partial\mathcal{D}_{2}(\mathsf{S}_{\ast})$ in the hardening limit as $\varepsilon\rightarrow 0$.
\end{itemize}
We now employ a convenient change of reference frame to study the dynamics of soft spheres, with a view to obtaining an explicit formula for $\sigma^{\varepsilon}$.
\subsection{Centre of Mass Reference Frame}
We follow the approach of  (\cite{MR3157048}, chapter 8, section 8.1) in reducing our study of the dynamics of two soft spheres to the centre-of-mass reference frame. As claimed above, this will make finding the explicit formulae for the scattering operators $\sigma^{\varepsilon}$ rather straightforward.

We now transform the system \eqref{softodes} above by defining new variables $y^{\varepsilon}:=x^{\varepsilon}-\ov{x}^{\varepsilon}$, $\ov{y}^{\varepsilon}:=\frac{1}{2}(x^{\varepsilon}+\ov{x}^{\varepsilon})$ and $w^{\varepsilon}:=v^{\varepsilon}-\ov{v}^{\varepsilon}$, $\ov{w}^{\varepsilon}:=\frac{1}{2}(v^{\varepsilon}+\ov{v}^{\varepsilon})$, which easily can be shown to satisfy the system of decoupled equations
\begin{equation}\label{com}
\frac{d}{dt}\left[
\begin{array}{c}
y^{\varepsilon} \\ w^{\varepsilon}
\end{array}
\right]=\left[
\begin{array}{c}
w^{\varepsilon} \\ -2\nabla\Phi^{\varepsilon}(y^{\varepsilon}) 
\end{array}
\right], \tag{S$^{\varepsilon}_{0}$} 
\end{equation}
and 
\begin{equation}\label{trivial}
\frac{d}{dt}\left[
\begin{array}{c}
\ov{y}^{\varepsilon} \\
\ov{w}^{\varepsilon}
\end{array}
\right]=\left[
\begin{array}{c}
\ov{w}^{\varepsilon} \\
0
\end{array}
\right].
\end{equation}
As solutions of system \eqref{com} and \eqref{trivial} are also unique for any given initial datum, they are in a bijective correspondence with solutions of \eqref{softodes}. We focus our attention on the unbarred system \eqref{com}. We notice that this system of equations also has a natural Hamiltonian structure given by the energy function
\begin{equation*}
H^{\varepsilon}_{0}(y, w):=\frac{1}{2}|w|^{2}+2\Phi^{\varepsilon}(y).
\end{equation*}
In particular, if $Z_{0}\in\Sigma^{-}$, one has that
\begin{equation*}
\frac{1}{2}|w^{\varepsilon}(t)|^{2}+2\Phi^{\varepsilon}(y^{\varepsilon}(t))=\frac{1}{2}|v_{0}-\ov{v}_{0}|^{2}
\end{equation*}
for all time $t\in\mathbb{R}$, which implies that
\begin{equation*}
y^{\varepsilon}(t)\in\left\{\eta\in\mathbb{R}^{3}\,:\,|\eta|\geq \Phi_{0}^{-1}\left(\frac{\varepsilon|v_{0}-\ov{v}_{0}|^{2}}{4}\right)\right\} \quad \text{for all}\hspace{2mm}t\in\mathbb{R}.
\end{equation*}
As such, the Hamiltonian structure of system \eqref{com} ensures there is a natural {\em distance of closest approach} for the centres of mass of the soft spheres, once an initial datum $Z_{0}$ (and therefore the total energy) for the dynamics has been fixed.
\subsection{Distance of Closest Approach of the Centres of Mass}
When $Z_{0}\in \Sigma^{-}$, it shall prove useful to obtain upper and lower bounds (in $\varepsilon$) on $\rho_{\ast}^{\varepsilon}=\rho_{\ast}^{\varepsilon}(Z_{0})>0$, the distance of closest approach of the centres of mass of the two soft spheres, defined by
\begin{equation*}
\rho_{\ast}^{\varepsilon}:=\min\left\{|x^{\varepsilon}(t)-\ov{x}^{\varepsilon}(t)|\,:\,t\in\mathbb{R}\right\}=\min\left\{|y^{\varepsilon}(t)|\,:\,t\in\mathbb{R}\right\}.
\end{equation*}
Together with the symmetry of solutions with respect to the apse line, the distance of closest approach $\rho_{\ast}^{\varepsilon}$ will allow us to estimate the difference $\Delta\tau^{\varepsilon}$ between the entrance and exit times in terms of the hardening parameter $\varepsilon$, which is crucial for obtaining our compactness result in $\bv$. In order to do this, we begin by making an observation on the time evolution of the angular momentum of solutions $y^{\varepsilon}(t)$ when measured with respect to the origin. Indeed, by spherical symmetry of the potential $\Phi^{\varepsilon}$, we find that
\begin{equation}\label{fakeam}
\frac{d}{dt}\left(y^{\varepsilon}(t)\wedge w^{\varepsilon}(t)\right)=0,
\end{equation}
in particular the value of $y^{\varepsilon}(t)\wedge w^{\varepsilon}(t)$ is fixed by the initial data $y_{0}$ and $w_{0}$. We use this observation to determine the {\em first} time $\tau_{\ast}^{\varepsilon}\geq 0$ for which $|y^{\varepsilon}(t)|$ is minimised, namely
\begin{equation}\label{tca}
\tau_{\ast}^{\varepsilon}:=\min\left\{t\geq 0\,:\,\dot{\rho}^{\varepsilon}(t)=0\right\}.
\end{equation} 
We separate our considerations into three cases.
\subsubsection{The Case $y_{0}\wedge w_{0}\neq 0$ and $y_{0}\cdot w_{0}\neq 0$} In this case, \eqref{fakeam} implies that $y^{\varepsilon}$ evolves for all time in the plane (that passes through the origin) which is orthogonal to the vector $y_{0}\wedge w_{0}$. We study its evolution with polar co-ordinates in this plane. Indeed, we may write
\begin{equation*}
y^{\varepsilon}(t)=R_{0}\left[
\begin{array}{c}
\rho^{\varepsilon}(t)e(\vartheta^{\varepsilon}(t)) \\
0
\end{array}
\right],
\end{equation*}
where $e(\theta):=(\sin\theta, \cos\theta)$, $\rho^{\varepsilon}(0)e(\vartheta^{\varepsilon}(0))=y_{0}$ and $R_{0}\in\mathrm{SO}(3)$ is the rotation matrix satisfying $R_{0}(0 ,0, 1)=y_{0}\wedge w_{0}$. From \eqref{fakeam}, we therefore find that
\begin{equation}\label{iden}
(\rho^{\varepsilon})^{4}(\dot{\vartheta}^{\varepsilon})^{2}=|y_{0}\wedge w_{0}|^{2}.
\end{equation}
Moreover, since the dynamics associated with system \eqref{com} conserves energy, we find using identity \eqref{iden} that
\begin{equation}\label{rode}
(\dot{\rho}^{\varepsilon})^{2}+\frac{A_{0}}{(\rho^{\varepsilon})^{2}}+\frac{4}{\varepsilon}\Phi_{0}(\rho^{\varepsilon})=2E_{0},
\end{equation}
with $A_{0}:=|y_{0}\wedge w_{0}|^{2}$ and $E_{0}:=\frac{1}{2}|w_{0}|^{2}$. As such, the radius $\rho^{\varepsilon}$ is at a minimum if and only if it satisfies the equation
\begin{equation}\label{critrad}
\frac{A_{0}}{(\rho^{\varepsilon})^{2}}+4\varepsilon^{-1}\Phi_{0}(\rho^{\varepsilon})=2E_{0}.
\end{equation}
We have the following simple lemma.
\begin{lem}
If $y_{0}=x_{0}-\ov{x}_{0}$ and $w_{0}=v_{0}-\ov{v}_{0}$ are such that $y_{0}\wedge w_{0}\neq 0$ and $y_{0}\cdot w_{0}\neq 0$, there exists a unique $0<\rho_{\ast}=\rho^{\varepsilon}_{\ast}(Z_{0})<1$ which satisfies the equation \eqref{critrad}.
\end{lem}
\begin{proof}
Since $\Phi^{\varepsilon}$ satisfies hypothesis (P1), namely that it is strictly decreasing on $(0, 1)$ and is of compact support on $\mathbb{R}$, we observe that the task of finding some $\rho_{\ast}>0$ that satisfies \eqref{critrad} is equivalent to proving that the map $r\mapsto 2E_{0}-A_{0}/r^{2}$ has a zero $r_{\ast}$ in the open interval $(0, 1)$. As every zero of this map is of the form $r_{\ast}=\sin\alpha$, where $\alpha=\alpha(y_{0}, w_{0})$ is given explicitly by $\alpha=\arcsin(\sqrt{A_{0}/2E_{0}})$, the claim of the lemma follows as it is neither the case that $y_{0}\cdot w_{0}=0$ nor $y_{0}\wedge w_{0}=0$.
\end{proof}
For such initial data $Z_{0}$, it therefore follows the the first time for which $|y^{\varepsilon}(t)|$ is minimised is $\tau_{\ast}^{\varepsilon}:=\mathrm{min}\{t\geq 0\,:\, \rho^{\varepsilon}(t) \hspace{2mm}\text{satisfies}\hspace{1mm}\eqref{critrad}\}$.
\subsubsection{The Case $y_{0}\wedge w_{0}=0$} In this case, it follows that $\rho^{\varepsilon}$ satisfies
\begin{equation*}
(\dot{\rho}^{\varepsilon})^{2}+4\Phi^{\varepsilon}_{0}(\rho^{\varepsilon})=2E_{0},
\end{equation*}
for all time, whence $\tau_{\ast}^{\varepsilon}:=\mathrm{min}\{t\geq 0\,:\, \rho^{\varepsilon}(t)=\Phi_{0}^{-1}(E_{0}\varepsilon/2)\}$
\subsubsection{The Case $y_{0}\cdot w_{0}=0$} Since we have the simple identity $\rho^{\varepsilon}(t)\dot{\rho}^{\varepsilon}(t)=y^{\varepsilon}(t)\cdot w^{\varepsilon}(t)$ for all $t\in\mathbb{R}$, it follows in this case that $\dot{\rho}^{\varepsilon}(0)=0$. As such, $\tau_{\ast}^{\varepsilon}=0$.
\subsection{Uniqueness of the Time of Closest Approach $\tau_{\ast}^{\varepsilon}$}
If we select our initial data $Z_{0}$ from $\Sigma^{-}$, it follows that $\dot{\rho}^{\varepsilon}(0)\leq 0$, namely $\rho^{\varepsilon}$ is non-increasing at time zero. However, $\rho^{\varepsilon}$ cannot decrease for all time $t>0$ due to the upper bound on the energy provided by \eqref{rode}. It is not immediate from our analysis so far (in the case that $y_{0}\cdot w_{0}\neq 0$) that there exists only one time that renders the time derivative $\dot{\rho}^{\varepsilon}$ zero, i.e. that $\mathrm{card}\left\{t\geq 0\,:\,\dot{\rho}^{\varepsilon}(t)=0\right\}=1$. The following lemma on symmetry of solutions $y^{\varepsilon}$ with respect to the apse line (which is a simple consequence of uniqueness of classical solutions of system \eqref{softodes}) yields the uniqueness of the {\em time of closest approach} $\tau_{\ast}^{\varepsilon}$.
\begin{lem}\label{symmetry}
Suppose $Z_{0}\in\Sigma^{-}$. For the associated classical solution $[y^{\varepsilon}, w^{\varepsilon}]$ of \eqref{com} and its corresponding first time of closest approach \eqref{tca}, it follows that
\begin{equation}\label{form}
Y^{\varepsilon}(t):=\left\{
\begin{array}{ll}
y^{\varepsilon}(t) & \quad \text{when}\hspace{2mm}t\leq \tau^{\varepsilon}_{\ast}, \vspace{2mm}\\
-\left(I-2\omega^{\varepsilon}_{\ast}\otimes \omega^{\varepsilon}_{\ast}\right)y^{\varepsilon}(2\tau^{\varepsilon}_{\ast}-t) & \quad \text{when}\hspace{2mm} t>\tau^{\varepsilon}_{\ast}.
\end{array}
\right.
\end{equation}
and $W^{\varepsilon}:=\dot{Y}^{\varepsilon}$ is also a classical solution of \eqref{com} on $\mathbb{R}$ with the same initial datum, where $\omega_{\ast}^{\varepsilon}\in\mathbb{S}^{2}$ is the {\bf apse line} given by
\begin{equation}\label{apseline}
\omega_{\ast}^{\varepsilon}:=R_{0}\left[
\begin{array}{c}
e(\vartheta_{\ast}^{\varepsilon})\\
0
\end{array}
\right]
\end{equation}
that corresponds to the {\em angle of deflection} $\vartheta_{\ast}^{\varepsilon}:=\vartheta^{\varepsilon}(\tau_{\ast}^{\varepsilon})$.
\end{lem}
\begin{proof}
Follows from a calculation that uses spherical symmetry of the potential $\Phi^{\varepsilon}$. We leave the details of this calculation to the reader.
\end{proof}
From the above deductions, we immediately yield the following useful information on the duration of collision $\Delta\tau^{\varepsilon}$.
\begin{cor}
If $Z_{0}\in\Sigma^{-}$, the duration of collision $\Delta\tau^{\varepsilon}(Z_{0})$ is given explicitly in terms of the time of closest approach $\tau_{\ast}^{\varepsilon}(Z_{0})$ by 
\begin{equation*}
\Delta\tau^{\varepsilon}(Z_{0})=\left\{
\begin{array}{ll}
0 & \quad \text{if}\hspace{2mm} y_{0}\cdot w_{0}=0, \vspace{2mm}\\
2\tau_{\ast}^{\varepsilon}(Z_{0}) & \quad \text{if}\hspace{2mm}y_{0}\cdot w_{0}\neq 0.
\end{array}
\right.
\end{equation*}
\end{cor}

We now have that the time of closest approach $\tau_{\ast}^{\varepsilon}(Z_{0})$ completely determines the duration of a collision between two soft spheres. For the purposes of obtaining the $\bv$ compactness result in section \ref{bvcomp}, one now needs to obtain upper bounds on $\tau^{\varepsilon}_{\ast}$ in $\varepsilon$. 
\subsection{Estimates on the Time of Closest Approach $\tau^{\varepsilon}_{\ast}$}\label{colltime}
From identity \eqref{rode} and the fact that $\dot{\rho}^{\varepsilon}$ experiences at most one sign change on $\mathbb{R}$, we deduce readily that the function $t\mapsto \rho^{\varepsilon}(t)$ satisfies the implicit equation
\begin{equation*}
\rho^{\varepsilon}(t)=\left\{
\begin{array}{ll}
1-\displaystyle \int_{0}^{t}\sqrt{2E_{0}-\frac{A_{0}}{(\rho^{\varepsilon}(s))^{2}}-4\varepsilon^{-1}\Phi_{0}(\rho^{\varepsilon}(s))}\,ds & \quad \text{if}\hspace{2mm}t\leq \tau_{\ast}^{\varepsilon} \vspace{2mm}\\
\rho_{\ast}^{\varepsilon}+\displaystyle\int_{\tau_{\ast}^{\varepsilon}}^{t}\sqrt{2E_{0}-\frac{A_{0}}{(\rho^{\varepsilon}(s))^{2}}-4\varepsilon^{-1}\Phi_{0}(\rho^{\varepsilon}(s))}\,ds & \quad \text{if}\hspace{2mm}t>\tau_{\ast}^{\varepsilon} 
\end{array}
\right.
\end{equation*}
This formula allows us to obtain the following exact expression for the time of closest approach $\tau_{\ast}^{\varepsilon}$ in terms of the potential $\Phi^{\varepsilon}$ and the initial datum $Z_{0}$ alone.
\begin{lem}
If $Z_{0}\in\Sigma^{-}$, it follows that
\begin{equation}\label{timetau}
\tau_{\ast}^{\varepsilon}(Z_{0})=\int_{\rho_{\ast}^{\varepsilon}}^{1}\frac{dr}{\sqrt{2E_{0}-\frac{A_{0}}{r^{2}}-4\varepsilon^{-1}\Phi_{0}(r)}}.
\end{equation}
\end{lem}
\begin{proof}
This follows from a simple application of the inverse function theorem to the function $\rho^{\varepsilon}$ on the time interval $(-\infty, \tau^{\varepsilon}_{\ast})$ and properties (P1) of the potential $\Phi_{0}$.
\end{proof}
We are now in a position to estimate the duration of collision $\Delta\tau^{\varepsilon}=\tau_{+}^{\varepsilon}-\tau_{-}^{\varepsilon}$.
\begin{prop}\label{estimation}
For $Z_{0}\in\Sigma^{-}\setminus\Sigma^{0}$, there exists a constant $C=C(Z_{0}, \Phi_{0})>0$ independent of the hardening parameter $\varepsilon$, and $\varepsilon_{0}=\varepsilon_{0}(Z_{0}, \Phi_{0})<1$, such that 
\begin{equation}\label{importantbound}
\tau_{\ast}^{\varepsilon}\leq C\varepsilon^{1/\beta}
\end{equation}
for all $0<\varepsilon\leq\varepsilon_{0}$.
\end{prop}
\begin{proof}
We must split the demonstration of this result into two cases.
\subsubsection*{Case I: $A_{0}\neq 0$}
We firstly obtain upper and lower bounds on the distance of closest approach $\rho_{\ast}^{\varepsilon}$. From identity \eqref{critrad}, it is clear that $4\varepsilon^{-1}\Phi_{0}(\rho_{\ast}^{\varepsilon})\leq 2E_{0}$, whence by the lower bound on $\Phi_{0}$ near 1 in (P1), we find that
\begin{equation}\label{lbr}
\rho_{\ast}^{\varepsilon}\geq 1-q_{1}\varepsilon^{1/\beta} \quad \text{where}\hspace{2mm}q_{1}:=\left(\frac{E_{0}}{2c_{1}}\right)^{1/\beta},
\end{equation}
which holds for all $0<\varepsilon<\varepsilon_{1}$, where $\varepsilon_{1}$ is determined by $\varepsilon_{1}:=\sup\{0<\varepsilon<1\,:\,r_{0}<\rho_{\ast}^{\varepsilon}\}$. On the other hand, together with this lower bound \eqref{lbr}, identity \eqref{critrad} implies that
\begin{equation*}
\Phi_{0}(\rho_{\ast}^{\varepsilon})\geq \frac{\varepsilon E_{0}}{2}-\frac{\varepsilon A_{0}}{4}\left(1-q_{1}\varepsilon^{\frac{1}{\beta}}\right)^{-2},
\end{equation*}
whence
\begin{equation}\label{ubr}
\rho_{\ast}^{\varepsilon}\leq 1-q_{2}(\varepsilon)\varepsilon^{1/\beta} \quad \text{where}\hspace{2mm}q_{2}(\varepsilon):=\left(\frac{1}{4c_{2}}\right)^{1/\beta}\left[2E_{0}-A_{0}\left(1-q_{1}\varepsilon^{1/\beta}\right)^{-2}\right]^{1/\beta},
\end{equation}
which holds for all $0<\varepsilon< \varepsilon_{2}$, where $\varepsilon_{2}$ is small enough so that
\begin{equation}\label{wabba}
\frac{2E_{0}-A_{0}}{2}<2E_{0}-A_{0}\left(1-q_{1}\varepsilon^{\frac{1}{\beta}}\right)^{-2}<2E_{0}-A_{0}.
\end{equation}
We now turn to bounding the integral \eqref{timetau}. Firstly, for any $0<\lambda<1$ we write $\tau_{\ast}^{\varepsilon}=I_{1}^{\varepsilon}(\lambda)+I_{2}^{\varepsilon}(\lambda)$, where
\begin{equation*}
I_{1}^{\varepsilon}(\lambda):=\int_{\rho_{\ast}^{\varepsilon}}^{\rho_{\ast}^{\varepsilon}(\lambda)}\frac{dr}{\sqrt{2E_{0}-\frac{A_{0}}{r^{2}}-4\varepsilon^{-1}\Phi_{0}(r)}} \quad \text{and}\quad I_{2}^{\varepsilon}(\lambda):=\int_{\rho_{\ast}^{\varepsilon}(\lambda)}^{1}\frac{dr}{\sqrt{2E_{0}-\frac{A_{0}}{r^{2}}-4\varepsilon^{-1}\Phi_{0}(r)}},
\end{equation*}
where $\rho_{\ast}^{\varepsilon}(\lambda):=\lambda+(1-\lambda)\rho_{\ast}^{\varepsilon}$. Now, by a change of measure in the integral $I_{1}^{\varepsilon}(\lambda)$, we find that
\begin{equation*}
I_{1}^{\varepsilon}\leq \frac{\sqrt{2E_{0}-A_{0}}}{4}|\Phi_{0}'(\rho_{\ast}^{\varepsilon}(\lambda))|^{-1}\varepsilon.
\end{equation*}
By employing the upper bound \eqref{ubr} on $\rho_{\ast}^{\varepsilon}$ and the lower bound on the derivative of $\Phi_{0}$ in (P3), we find that
\begin{equation*}
I_{1}^{\varepsilon}\leq \frac{\sqrt{2E_{0}-A_{0}}(1-\lambda)^{1-\beta}q_{2}(\varepsilon)^{1-\beta}}{4\kappa_{1}}\varepsilon^{1/\beta} \quad \text{for all}\hspace{2mm}0<\varepsilon<\varepsilon_{2}.
\end{equation*}
Since $2E_{0}-A_{0}>0$, it follows from \eqref{wabba} that
\begin{equation}\label{helpu}
\left[\frac{2E_{0}-A_{0}}{8c_{2}}\right]^{1/\beta}<q_{2}(\varepsilon)<\left[\frac{2E_{0}-A_{0}}{4c_{2}}\right]^{1/\beta}
\end{equation}
for all $0<\varepsilon<\varepsilon_{2}$. Thus, there exists a constant $C_{1}=C_{1}(Z_{0}, \Phi_{0}, \lambda)>0$ given explicitly by
\begin{equation*}
C_{1}:=\frac{\sqrt{2E_{0}-A_{0}}(1-\lambda)^{1-\beta}}{8\kappa_{1}}\left[\frac{2E_{0}-A_{0}}{8c_{2}}\right]^{(1-\beta)/\beta}
\end{equation*}
and $\varepsilon_{3}:=\min\{\varepsilon_{1}, \varepsilon_{2}\}$ such that
\begin{equation}\label{ione}
I_{1}^{\varepsilon}\leq C_{1}\varepsilon^{1/\beta}
\end{equation}
for all $0<\varepsilon<\varepsilon_{3}$. 

We now estimate the second integral contributing to $\tau_{\ast}^{\varepsilon}$. Indeed, we have
\begin{equation*}
I_{2}^{\varepsilon}\leq \max_{\rho_{\varepsilon}(\lambda)\leq r\leq 1}\frac{1-\lambda}{\sqrt{2E_{0}-\frac{A_{0}}{r^{2}}-2\Phi_{0}^{\varepsilon}(r)}}(1-\rho_{\ast}^{\varepsilon}),
\end{equation*}
which from monotonicity of $\Phi_{0}$ and the lower bound \eqref{lbr} on $\rho_{\ast}^{\varepsilon}$ yields the bound
\begin{equation*}
I_{2}^{\varepsilon}\leq (1-\lambda)\left[\underbrace{2E_{0}-A_{0}\rho_{\ast}^{\varepsilon}(\lambda)^{-2}-4\varepsilon^{-1}\Phi_{0}(\rho_{\ast}^{\varepsilon}(\lambda))}_{C_{2}(\varepsilon):=}\right]^{-1/2}q_{1}\varepsilon^{1/\beta}.
\end{equation*}
We note that
\begin{equation*}
\left[c_{1}-c_{2}(1-\lambda)^{\beta}\right](1-\rho_{\ast}^{\varepsilon})^{\beta}\leq \Phi_{0}(\rho_{\ast}^{\varepsilon})-\Phi_{0}(\rho_{\ast}^{\varepsilon}(\lambda)).
\end{equation*}
Using the lower bound in \eqref{helpu}, we infer that
\begin{equation*}
\left(\frac{2E_{0}-A_{0}}{8c_{2}}\right)\left[c_{1}-c_{2}(1-\lambda)^{\beta}\right]\varepsilon\leq \Phi_{0}(\rho_{\ast}^{\varepsilon})-\Phi_{0}(\rho_{\ast}^{\varepsilon}(\lambda)),
\end{equation*}
whence finally choosing $\lambda=\lambda_{0}$ to satisfy the inequality $(1-\lambda_{0})^{\beta}<c_{1}/c_{2}$, we obtain that
\begin{equation}\label{vhelpu}
\lim_{\varepsilon\rightarrow 0}\varepsilon^{-1}\Phi_{0}(\rho_{\ast}^{\varepsilon})>\lim_{\varepsilon\rightarrow 0}\varepsilon^{-1}\Phi_{0}(\rho_{\ast}^{\varepsilon}(\lambda_{0})).
\end{equation}
At this point, we note from identity \eqref{critrad} that $\lim_{\varepsilon\rightarrow 0}\varepsilon^{-1}\Phi_{0}(\rho_{\ast}^{\varepsilon})=E_{0}/2-A_{0}/4$. By \eqref{vhelpu}, it follows that
\begin{equation*}
\lim_{\varepsilon\rightarrow 0}\varepsilon^{-1}\Phi_{0}(\rho_{\ast}^{\varepsilon}(\lambda_{0}))<\frac{E_{0}}{2}-\frac{A_{0}}{4},
\end{equation*}
whence $C_{2}(\varepsilon)$ is bounded strictly away from $0$ for $\varepsilon$ in a sufficiently-small neighbourhood of 0, say $0<\varepsilon<\varepsilon_{4}$. As such, we obtain the bound
\begin{equation}\label{itwo}
I_{2}^{\varepsilon}\leq C_{2}\varepsilon^{1/\beta},
\end{equation}
for $0<\varepsilon<\varepsilon_{4}$. The claim of the proposition follows from estimates \eqref{ione} and \eqref{itwo} and setting $\varepsilon_{0}:=\min\{\varepsilon_{3}, \varepsilon_{4}\}$.
\subsubsection*{Case II: $A_{0}=0$}
We have that $4\varepsilon^{-1}\Phi_{0}(\rho_{\ast}^{\varepsilon})=2E_{0}$, whence by (P2) we infer that
\begin{equation*}
1-\left(\frac{E_{0}\varepsilon}{2c_{2}}\right)^{1/\beta}\leq \rho_{\ast}^{\varepsilon}\leq 1-\left(\frac{E_{0}\varepsilon}{2c_{1}}\right)^{1/\beta}.
\end{equation*}
By considerations similar to case I above, we find that there exists a constant $C=C(Z_{0}, \Phi_{0})$ independent of the hardening parameter $\varepsilon$ such that $\tau_{\ast}^{\varepsilon}\leq C\varepsilon^{1/\beta}$ for all $0<\varepsilon<\varepsilon_{0}$, for some threshold $\varepsilon_{0}<1$.
\end{proof}
\subsection{Construction and Limiting Behaviour of the Scattering Operators $\sigma^{\varepsilon}$}
We have now done enough work to write down an explicit expression for the scattering operator $\sigma^{\varepsilon}:\Sigma^{-}\rightarrow\Sigma^{+}$ defined in \ref{defscat} above.
\begin{prop}
The scattering map $\sigma^{\varepsilon}$ is given explicitly by the components
\begin{equation*}
\Pi_{1}\sigma^{\varepsilon}Z_{0}=\left(\left[
\begin{array}{cc}
0 & I \\
I & 0
\end{array}
\right]+\left[
\begin{array}{cc}
\omega_{\ast}^{\varepsilon}\otimes \omega_{\ast}^{\varepsilon} & -\omega_{\ast}^{\varepsilon}\otimes \omega_{\ast}^{\varepsilon} \\
-\omega_{\ast}^{\varepsilon}\otimes \omega_{\ast}^{\varepsilon} & \omega_{\ast}^{\varepsilon}\otimes \omega_{\ast}^{\varepsilon}
\end{array}
\right]\right)\Pi_{1}Z_{0}+\tau_{\ast}^{\varepsilon}\left[
\begin{array}{cc}
I & I \\
I & I
\end{array}
\right]\Pi_{2}Z_{0}
\end{equation*}
and
\begin{equation*}
\Pi_{2}\sigma^{\varepsilon}Z_{0}=(I-2\widehat{\nu}_{\ast}^{\varepsilon}\otimes \widehat{\nu}_{\ast}^{\varepsilon})\Pi_{2}Z_{0},
\end{equation*}
where $\widehat{\nu}_{\ast}^{\varepsilon}\in\mathbb{S}^{5}$ denotes
\begin{equation*}
\widehat{\nu}_{\ast}^{\varepsilon}:=\frac{1}{\sqrt{2}}\left[
\begin{array}{c}
\omega_{\ast}^{\varepsilon} \\
-\omega_{\ast}^{\varepsilon}
\end{array}
\right].
\end{equation*}
\end{prop}
\begin{proof}
This follows from a change of co-ordinates from and the structural formula \eqref{form} for the evolution of the reduced system $(y^{\varepsilon}, w^{\varepsilon})$. 
\end{proof}
It will also be of use to characterise the limiting behaviour of the apse line as $\varepsilon\rightarrow 0$. As usual, we break our considerations (corresponding to the choice of initial datum $Z_{0}\in\Sigma^{-}$) into three cases. We note firstly that the apse line is particularly simple in the case when $Z_{0}\in\Sigma^{-}$ satisfies $y_{0}\wedge w_{0}=0$. Indeed, from identity \eqref{fakeam} it follows that $\vartheta^{\varepsilon}(t)=\vartheta_{0}$ for all $t=0$, and so
\begin{equation*}
\omega_{\ast}^{\varepsilon}=R_{0}\left[\
\begin{array}{c}
e(\vartheta_{0})\\
0
\end{array}
\right].
\end{equation*}
On the other hand, if $Z_{0}\in\Sigma^{-}$ is such that $y_{0}\cdot w_{0}=0$, then $\rho^{\varepsilon}(t)$ is minimised at $t=0$ and 
\begin{equation*}
\frac{dv^{\varepsilon}}{dt}=0\quad \text{and}\quad \frac{d\ov{v}^{\varepsilon}}{dt}=0\quad \text{for all}\hspace{2mm}t\in\mathbb{R},
\end{equation*}
whence $\omega_{\ast}^{\varepsilon}=R_{0}[e(\vartheta_{0}), 0]$. We now consider the last case.
\begin{lem}
For any $Z_{0}\in\Sigma^{-}$ such that $y_{0}\wedge w_{0}\neq 0$ and $y_{0}\cdot w_{0}\neq 0$, one has that
\begin{equation*}
\omega_{\ast}^{\varepsilon}\rightarrow x_{0}-\ov{x}_{0}\quad \text{as}\hspace{2mm}\varepsilon\rightarrow 0.
\end{equation*}
\end{lem}
\begin{proof}
We begin by noting that the maps $\rho^{\varepsilon}, \vartheta^{\varepsilon}$ satisfy the identity
\begin{equation*}
R_{0}\left[
\begin{array}{c}
0 \\
0 \\
(\rho^{\varepsilon}(t))^{2}\dot{\vartheta}^{\varepsilon}(t)
\end{array}
\right]=y_{0}\wedge w_{0}
\end{equation*}
for all time. Since $t\mapsto(\rho^{\varepsilon}(t), \vartheta^{\varepsilon}(t))$ parametrises a $C^{1}$ polar curve in $\mathbb{R}^{2}$, an application of the chain rule yields the explicit representation formula for the deflection angle $\vartheta_{\ast}^{\varepsilon}$,
\begin{equation*}
\vartheta_{\ast}^{\varepsilon}=\vartheta_{0}+\int_{\rho_{\ast}^{\varepsilon}}^{1}\frac{R_{0}^{T}(y_{0}\wedge w_{0})\cdot e_{3}}{r^{2}\sqrt{2E_{0}-\frac{A_{0}}{r^{2}}-4\varepsilon^{-1}\Phi_{0}(r)}}\,dr
\end{equation*}
where $e_{3}:=(0, 0, 1)$ and $\vartheta_{0}$ satisfies $e(\vartheta_{0})=y_{0}$. Using estimates similar to those in the proof of proposition \ref{estimation}, we find that $\lim_{\varepsilon\rightarrow 0}\vartheta_{\ast}^{\varepsilon}=\vartheta_{0}$. Thus, we infer that $\omega_{\ast}^{\varepsilon}\rightarrow y_{0}$ as $\varepsilon\rightarrow 0$, and so the limiting apse line is determined by the initial spatial data $x_{0}$ and $\ov{x}_{0}$ alone.
\end{proof}
We are now in a position to characterise the limiting form of the scattering operators $\sigma^{\varepsilon}$.
\begin{cor}\label{compactcon}
For any $Z_{0}\in\mathcal{D}_{2}(\mathsf{S}_{\ast})$, we have
\begin{equation*}
\lim_{\varepsilon\rightarrow 0}(\Pi_{1}\sigma^{\varepsilon})(Z_{0})=Z_{0} \quad \text{and}\quad \lim_{\varepsilon\rightarrow 0}(\Pi_{2}\sigma^{\varepsilon})(Z_{0})=\sigma_{n}(\Pi_{2}Z_{0}),
\end{equation*}
where $\sigma_{n}:\mathbb{R}^{6}\rightarrow\mathbb{R}^{6}$ is the classical {\em Boltzmann scattering matrix} given by
\begin{equation*}
\sigma_{n}:=I-2\widehat{\nu}_{n}\otimes \widehat{\nu}_{n},
\end{equation*}
where $\widehat{\nu}_{n}$ is given by
\begin{equation*}
\widehat{\nu}_{n}:=\frac{1}{\sqrt{2}}\left[
\begin{array}{c}
n \\
-n
\end{array}
\right], \quad \text{with}\quad n:=\ov{x}_{0}-x_{0}.
\end{equation*}
\end{cor}
\begin{proof}
We leave the straightforward details to the reader.
\end{proof}
As such, in the hardening limit $\varepsilon\rightarrow 0$ the change of velocity is determined by the classical Boltzmann scattering matrix.
\section{Compactness in $\mathrm{BV}_{\mathrm{loc}}(\mathbb{R}, \mathbb{R}^{6})$}\label{bvcomp}
With the bound \eqref{importantbound} on the time of collision in terms of the hardening parameter $\varepsilon>0$ in hand, it is natural to ask in what space one should aim for compactness for the sequences of smooth trajectories $\{Z^{\varepsilon}\}_{0<\varepsilon<1}$. It has already been shown that classical solutions $t\mapsto Z(t)$ to system \eqref{sminus} and \eqref{splus} are unique, with the velocity maps being of the form
\begin{equation}\label{vclass}
v(t)=\left\{
\begin{array}{ll}
v_{0}&\quad \text{for}\hspace{2mm}t\leq 0, \vspace{2mm}\\
v_{0}-[(v_{0}-\ov{v}_{0})\cdot n] n & \quad \text{for}\hspace{2mm}t>0,
\end{array}
\right.
\end{equation}
and
\begin{equation}\label{ovvclass}
\ov{v}(t)=\left\{
\begin{array}{ll}
\ov{v}_{0}&\quad \text{for}\hspace{2mm}t\leq 0, \vspace{2mm}\\
\ov{v}_{0}+[(v_{0}-\ov{v}_{0})\cdot n] n & \quad \text{for}\hspace{2mm}t>0.
\end{array}
\right.
\end{equation}
These maps (which are to be regarded as limit maps of $v^{\varepsilon}$ and $\ov{v}^{\varepsilon}$, respectively) are evidently of locally bounded variation on $\mathbb{R}$. Indeed, we shall use the estimate \eqref{timetau} to establish suitable bounds on $v^{\varepsilon}$ and $\ov{v}^{\varepsilon}$ on bounded intervals of time. In what follows, we make use of some basic results in the theory of $\mathrm{BV}_{\mathrm{loc}}(\mathbb{R}, \mathbb{R}^{6})$ maps (which one can find in the book of \textsc{Ambrosio, Fusco and Pallara} \cite{MR1857292} or \textsc{Evans and Gariepy} \cite{MR1158660}, for instance). For the convenience of the reader, we recall a few basic definitions and results.
\begin{defn}[Functions of Bounded Variation]
Suppose $U\subseteq\mathbb{R}$ is an open set and $M\geq 1$. We say that $u=(u_{1}, ..., u_{M})\in L^{1}(U, \mathbb{R}^{M})$ is of {\bf bounded variation on $U$} if and only if there exist finite Radon measures $Du_{i}$ on $U$ ($i=1, ..., M$) such that
\begin{equation*}
\int_{U}u_{i}\phi'\,dt=-\int_{U}\phi \,dDu_{i} \quad \text{for all}\hspace{2mm}\phi\in C^{1}_{0}(U, \mathbb{R}),
\end{equation*}
for $i=1, ..., M$, i.e. $Du_{i}$ is the distributional derivative of $u_{i}$. The vector space of all such maps is denoted $\mathrm{BV}(U, \mathbb{R}^{M})$.
\end{defn}
\begin{defn}[Variation]
Suppose $U\subseteq\mathbb{R}$ is an open set. For $u\in L^{1}_{\mathrm{loc}}(U, \mathbb{R}^{M})$, the {\bf variation} $\mathrm{Var}(u, U)$ of $u$ on $U$ is defined by
\begin{equation*}
\mathrm{Var}(u, U):=\sup\left\{\int_{U}u\cdot\varphi'\,dt\,:\,\varphi\in C^{1}_{0}(U, \mathbb{R}^{M})\hspace{2mm}\text{with}\hspace{2mm}\|\varphi\|_{L^{\infty}(U)}\leq 1\right\}.
\end{equation*}
\end{defn}
It is straightforward to show that a given $u\in L^{1}(U, \mathbb{R}^{M})$ lies in $\mathrm{BV}(U, \mathbb{R}^{M})$ if and only if $\mathrm{Var}(u, U)<\infty$. As mentioned above, we shall also employ maps of locally bounded variation on $\mathbb{R}$.
\begin{defn}[Functions of Locally Bounded Variation]
Suppose $U\subseteq\mathbb{R}$ is an open set and $M\geq 1$. A map $u=(u_{1}, ..., u_{M})\in L^{1}_{\mathrm{loc}}(U, \mathbb{R}^{M})$ is said to be of {\bf locally bounded variation on $U$} if and only if $\mathrm{Var}(u, W)<\infty$ for all open subsets $W$ compactly contained in $U$. The vector space of all such maps is denoted $\mathrm{BV}_{\mathrm{loc}}(U, \mathbb{R}^{M})$.
\end{defn}
From a computational point of view, the variation $\mathrm{Var}(u, U)$ is not particularly convenient in this article. We shall instead work with an equivalent notion of {\em pointwise variation} of maps on $U$. We refer the reader to (\cite{MR1857292}, chapter 3) for full details.
\begin{defn}[Pointwise and Essential Variations]
Let $I=(a, b)$ be an open subinterval of $\mathbb{R}$, and suppose $u: I\rightarrow\mathbb{R}^{M}$ is any map defined on $I$. The {\bf pointwise variation $\mathrm{pVar}(u, I)$ of $u$ over $I$} is defined by
\begin{equation*}
\mathrm{pVar}(u, I):=\sup\left\{ \mathrm{Var}(u, I; \mathcal{P})\,:\,\mathcal{P}\hspace{2mm}\text{is a partition of $I$}\right\},
\end{equation*}
where 
\begin{equation*}
\mathrm{Var}(u, I; \mathcal{P}):=\sum_{i=1}^{N-1}|u(t_{i+1})-u(t_{i})|\quad \text{when}\hspace{2mm}\mathcal{P}=\{t_{i}\}_{i=1}^{N}.
\end{equation*}
On the other hand, the {\bf essential variation $\mathrm{eVar}(u, I)$ of $u$ over $I$} is defined by
\begin{equation*}
\mathrm{eVar}(u, I):=\inf\left\{pV(v, I)\,:\, v=u \hspace{2mm}\text{almost everywhere on}\hspace{2mm}I\right\}.
\end{equation*}
\end{defn}
We shall use in all the sequel the fact that for any map $u\in L^{1}(I, \mathbb{R}^{M})\cap BV(I, \mathbb{R}^{M})$, $\mathrm{Var}(u, I)=\mathrm{eVar}(u, I)$. Moreover, $BV(I, \mathbb{R}^{M})$ admits the structure of a Banach space with respect to the norm
\begin{equation*}
\|u\|_{\mathrm{BV}(I, \mathbb{R}^{M})}:=\|u\|_{L^{1}(I)}+\mathrm{Var}(u, I).
\end{equation*}
We now quote from a compactness result for sequences of norm-bounded maps.
\begin{prop}\label{comp1}
Suppose $\{u_{j}\}_{j=1}^{\infty}\subset \mathrm{BV}(I, \mathbb{R}^{M})$ is a sequence which is uniformly bounded in norm, i.e. $\|u_{j}\|_{\mathrm{BV}(I, \mathbb{R}^{M})}<C$ for all $j\geq 1$. There exists a subsequence $\{u_{j(k)}\}_{k=1}^{\infty}$ and a map $u\in\mathrm{BV}(I, \mathbb{R}^{M})$ such that $u_{j(k)}\rightarrow u$ in $L^{1}(I)$ as $k\rightarrow \infty$.
\end{prop}
Rather than the norm topology, it will be convenient for us to work with the weak-$\ast$ topology on $\mathrm{BV}(I, \mathbb{R}^{6})$ instead.
\begin{defn}[Weak-$\ast$ Convergence in $\mathrm{BV}(I, \mathbb{R}^{M})$]
Suppose that $u^{\varepsilon}, u\in\mathrm{BV(I, \mathbb{R}^{M})}$. We say that $u^{\varepsilon}$ converges to $u$ in $\mathrm{BV}(I, \mathbb{R}^{M})$ in the weak-$\ast$ topology if and only if $u^{\varepsilon}\rightarrow u$ in $L^{1}(I, \mathbb{R}^{M})$ as $\varepsilon\rightarrow 0$, and 
\begin{equation*}
\lim_{\varepsilon\rightarrow 0}\int_{a}^{b}u^{\varepsilon}_{i}\phi'\,dt=-\int_{a}^{b}\phi \,dDu_{i} \quad \text{for all} \hspace{2mm}\phi\in C(I)
\end{equation*}
for $i=1, ..., M$.
\end{defn}
We finish this preliminary section with the following basic weak-$\ast$ compactness result for $\mathrm{BV}(I, \mathbb{R}^{M})$.
\begin{prop}[Weak-$\ast$ Compactness in $\mathrm{BV}(I, \mathbb{R}^{M})$]\label{weakstar}
A family of maps $\{u^{\varepsilon}\}_{0<\varepsilon<1}$ converges weakly-$\ast$ in $\mathrm{BV}(I, \mathbb{R}^{M})$ to $u\in\mathrm{BV}(I, \mathbb{R}^{M})$ if and only if $\{u^{\varepsilon}\}_{0<\varepsilon<1}$ is bounded in $\mathrm{BV}(I, \mathbb{R}^{M})$ and $u^{\varepsilon}$ converges strongly to $u$ in $L^{1}(I, \mathbb{R}^{M})$ as $\varepsilon\rightarrow 0$. 
\end{prop}
Our strategy in the following section will be to prove, using the uniform bound \eqref{importantbound}, that $v^{\varepsilon}, \ov{v}^{\varepsilon}\in\mathrm{BV}_{\mathrm{loc}}(\mathbb{R}, \mathbb{R}^{3})$ have the property that $\|V^{\varepsilon}\|_{\mathrm{BV}(I, \mathbb{R}^{6})}\leq C(I)$ for any open subinterval $I\subset\mathbb{R}$, where $V^{\varepsilon}=[v^{\varepsilon}, \ov{v}^{\varepsilon}]$. In turn, we shall be able to pass from classical solutions of system \eqref{softodes} to weak solutions of system \eqref{sminus} and \eqref{splus} in the limit as $\varepsilon\rightarrow 0$.
\begin{rem}
To our knowledge, there is currently no global-in-time existence and regularity theory in the literature for the analogous equations of motion for compact, strictly-convex, non-spherical particles whose boundary surfaces are of class $C^{1}$. Indeed, it may be possible that infinitely-many collisions of two non-spherical particles in a finite time interval take place for a given initial datum. Were this the case, it would not be that all velocity maps lie in $\mathrm{BV}_{\mathrm{loc}}(\mathbb{R})$, leading one to establish a compactness principle for approximate trajectories in another suitable functional space. We return to these remarks in the final section of the article.
\end{rem}
\subsection{Construction of Uniform Bounds}
As intimated above, rather than working with the variation $\mathrm{Var}(\cdot, I)$ over open subintervals $I$ of the real line, due to the structure of the equations of motion \eqref{softodes} it will be much more convenient to work with the (essential) pointwise variation $\mathrm{eVar}(\cdot, I)$. We begin by noticing that $\|\cdot\|_{\mathrm{BV}(I, \mathbb{R}^{6})}$-bounds on trajectories $V^{\varepsilon}$ for which $Z_{0}\in\Sigma^{0}$ are trivial. Indeed, it follows that the associated unique solution $Z^{\varepsilon}$ of \eqref{softodes} satisfies $dv^{\varepsilon}/dt=0$ and $d\ov{v}^{\varepsilon}/dt=0$ for all time $t\in\mathbb{R}$ and so are $\varepsilon$-independent. 

Let us now consider the non-trivial case when $Z_{0}\in \Sigma^{-}\setminus\Sigma^{0}$. As it will be useful when computing the variation $\mathrm{eVar}(\cdot, I)$ of velocity profiles, for any two times $t_{j}<t_{j+1}$ we notice (in the case of the velocity profile $v^{\varepsilon}$) that
\begin{equation}\label{possy}
v^{\varepsilon}(t_{j+1})-v^{\varepsilon}(t_{j})=\left\{
\begin{array}{ll}
0 & \quad t_{j}<t_{j+1}<0, \vspace{2mm}\\ 
-\int_{0}^{t_{j+1}}\nabla\Phi^{\varepsilon}(x^{\varepsilon}(s)-\ov{x}^{\varepsilon}(s))\,ds & \quad t_{j}<0<t_{j+1}\leq \tau_{+}^{\varepsilon}, \vspace{2mm}\\
-\int_{0}^{\tau_{+}^{\varepsilon}}\nabla\Phi^{\varepsilon}(x^{\varepsilon}(s)-\ov{x}^{\varepsilon}(s))\,ds & \quad t_{j}<0<\tau_{+}^{\varepsilon}<t_{j+1}, \vspace{2mm} \\ 
-\int_{t_{j}}^{t_{j+1}}\nabla\Phi^{\varepsilon}(x^{\varepsilon}(s)-\ov{x}^{\varepsilon}(s))\,ds & \quad 0\leq t_{j}<t_{j+1}\leq \tau_{+}^{\varepsilon}, \vspace{2mm}\\
-\int_{t_{j}}^{\tau_{+}^{\varepsilon}}\nabla\Phi^{\varepsilon}(x^{\varepsilon}(s)-\ov{x}^{\varepsilon}(s))\,ds & \quad 0\leq t_{j}<\tau_{+}^{\varepsilon}<t_{j+1}, \vspace{2mm}\\
0 & \quad \tau_{+}^{\varepsilon}<t_{j}<t_{j+1}.
\end{array}
\right.
\end{equation}
To start, we have the following auxiliary lemma.
\begin{lem}\label{lemm}
For any $T_{0}<T_{1}$, any partition $\mathcal{P}$ of $[T_{0}, T_{1}]$ and any $Z_{0}\in\Sigma^{-}$, there exists a constant $C=C(Z_{0}, \Phi_{0}, T_{0}, T_{1})>0$ independent of the partition $\mathcal{P}$ and the hardening parameter $\varepsilon>0$ such that
\begin{equation}\label{uniformb}
\mathrm{Var}(V^{\varepsilon}, I; \mathcal{P})\leq C \quad \text{with}\hspace{2mm}I:=(T_{0}, T_{1}),
\end{equation}
for all $0<\varepsilon<\varepsilon_{\ast}\leq1$, where $\varepsilon_{\ast}=\varepsilon_{\ast}(Z_{0}, \Phi_{0}, T_{0}, T_{1})$.
\end{lem}
\begin{proof}
The demonstration of this result follows from a careful case-by-case analysis. We do not demonstrate here the validity of the uniform bound \eqref{uniformb} in all cases that require consideration, i.e. for all possible choices of initial data $Z_{0}\in\Sigma^{-}$, open subintervals $I:=(T_{0}, T_{1})\subset\mathbb{R}$ and partitions $\mathcal{P}$ of $I$. We shall only establish the bound in the most involved case, and leave the proof of the other simpler cases to the reader.

Let us suppose, for instance, that $T_{0}<0<T_{1}$, i.e. the interval $I$ contains the collision time $t=0$. For $\varepsilon>0$ small enough, one has that $\tau_{+}^{\varepsilon}=2\tau_{\ast}^{\varepsilon}<T_{1}$. It then follows quickly using \eqref{possy} that
\begin{equation*}
\mathrm{Var}(V^{\varepsilon}, I; \mathcal{P})
\leq 2\int_{0}^{2\tau_{\ast}^{\varepsilon}}|\nabla\Phi^{\varepsilon}(x^{\varepsilon}(t)-\ov{x}^{\varepsilon}(t))|\,dt
\end{equation*}
for any non-trivial partition $\mathcal{P}=\{t_{j}\}_{j=1}^{N(\mathcal{P})}$ of the interval $(T_{0}, T_{1})$. By spherical symmetry of the potential $\Phi^{\varepsilon}$, we infer that
\begin{equation*}
\mathrm{Var}(V^{\varepsilon}, I; \mathcal{P})\leq 4\varepsilon^{-1}\tau_{\ast}^{\varepsilon}\max_{0\leq t \leq 2\tau_{\ast}^{\varepsilon}}|\Phi_{0}'(\rho^{\varepsilon}(t))|,
\end{equation*}
whence by the estimate in proposition \ref{estimation} for the exit time $\tau_{+}^{\varepsilon}=2\tau_{\ast}^{\varepsilon}$ together with the upper bound (P3) on $|\Phi_{0}'(r)|$ when $r$ is close to 1, we find that
\begin{equation*}
\mathrm{Var}(V^{\varepsilon}, I; \mathcal{P})\leq C\kappa_{2}q_{1}^{\beta-1}
\end{equation*}
for $\varepsilon<1$ sufficiently small, where $C>0$ is the constant in \eqref{importantbound}.
\end{proof}
With this in place, we have the following corollary.
\begin{prop}\label{foosh}
For a given open interval $I$, there exist a sequence $\{V_{i}\}_{i=1}^{\infty}\subset\{V^{\varepsilon}\}_{0<\varepsilon<1}$ and $V_{\infty}:=[v_{\infty}, \ov{v}_{\infty}]\in\mathrm{BV}_{\mathrm{loc}}(\mathbb{R}, \mathbb{R}^{6})$ such that $V_{i}\rightarrow V_{\infty}$ in $L^{1}(I)$ as $i\rightarrow\infty$ for any open interval $I\subset\mathbb{R}$.
\end{prop}
\begin{proof}
Via another case-by-case analysis, it is possible to show that $\|V^{\varepsilon}\|_{L^{1}(I, \mathbb{R}^{6})}\leq C'$ for some $C'=C'(Z_{0}, \Phi_{0}, T_{0}, T_{1})$ and all $\varepsilon>0$ sufficiently small. An application of \ref{comp1} yields the proof of the proposition.
\end{proof}
By a direct application of proposition \ref{weakstar}, proposition \ref{foosh} finally leads us to:
\begin{cor}
Suppose $Z_{0}\in\Sigma^{-}\setminus\Sigma^{0}$. For any $n\geq 1$, there exists a sequence $\{V_{i}\}_{i=1}^{\infty}\subset\{V^{\varepsilon}\}_{0<\varepsilon<1}$ such that $V_{i}\overset{\ast}{\rightharpoonup}V_{\infty}$ in $\mathrm{BV}((-n, n), \mathbb{R}^{6})$ as $i\rightarrow\infty$.
\end{cor}
In the following section, we shall in fact show that the whole family $\{V^{\varepsilon}\}_{0<\varepsilon<1}$ converges to the unique classical trajectory $V=[v, \ov{v}]$ as $\varepsilon\rightarrow 0$ (where $v$ and $\ov{v}$ are given by \eqref{vclass} and \eqref{ovvclass}, respectively). Let us now proceed to the proof of theorem \ref{mainresy}.
\subsection{Weak Solutions of the Hard Sphere System}
Making use of the compactness results of the previous section, it is now straightforward to pass from classical solutions of \eqref{softodes} to weak solutions of \eqref{sminus} and \eqref{splus}.

{\em Proof of Theorem \ref{mainresy}.} For any $\psi\in C^{1}_{0}(\mathbb{R}, \mathbb{R}^{6})$, it holds that
\begin{equation*}
\int_{-\infty}^{\infty}\frac{d}{dt}\left[
\begin{array}{c}
v^{\varepsilon} \\
\ov{v}^{\varepsilon}
\end{array}
\right]\cdot\psi\,dt=-\int_{-\infty}^{\infty}\left[
\begin{array}{c}
v^{\varepsilon}\\
\ov{v}^{\varepsilon}
\end{array}
\right]\cdot \psi'\,dt,
\end{equation*} 
and by passing to limits as $\varepsilon\rightarrow 0$, we infer that
\begin{equation*}
\int_{-\infty}^{\infty}\left[
\begin{array}{c}
v_{\infty}\\
\ov{v}_{\infty}
\end{array}
\right]\cdot\psi'\,dt=-\int_{-\infty}^{\infty}\psi\,dDV_{\infty},
\end{equation*} 
where $V_{\infty}:=[v_{\infty}, \ov{v}_{\infty}]$ is the limit point guaranteed by proposition \ref{foosh}. Moreover, since $x^{\varepsilon}(t)=x_{0}+tv(t)$, using the strong convergence of $[v^{\varepsilon}, \ov{v}^{\varepsilon}]$ to $[v_{\infty}, \ov{v}_{\infty}]$ in $L^{1}(I, \mathbb{R}^{6})$ as $\varepsilon\rightarrow 0$, we infer that $x_{\infty}:=x_{0}+tv_{\infty}(t)$ and $\ov{x}_{\infty}:=\ov{x}_{0}+t\ov{v}_{\infty}(t)$ satisfy
\begin{equation*}
\int_{-\infty}^{\infty}\left[
\begin{array}{c}
x_{\infty}\\
\ov{x}_{\infty}
\end{array}
\right]\cdot \varphi'\,dt=-\int_{-\infty}^{\infty}\left[
\begin{array}{c}
v_{\infty}\\
\ov{v}_{\infty}
\end{array}
\right]\cdot\varphi\,dt,
\end{equation*}
for any $\varphi\in C^{1}_{0}(\mathbb{R}, \mathbb{R}^{6})$. We must now pass comment on the conservation laws. For the conservation of linear momentum, one has
\begin{equation*}
v^{\varepsilon}(t)+\ov{v}^{\varepsilon}(t)=v_{0}+\ov{v}_{0}\quad \text{for all}\hspace{2mm}t\in\mathbb{R},
\end{equation*}
whence
\begin{equation*}
v_{\infty}(t)+\ov{v}_{\infty}(t)=v_{0}+\ov{v}_{0}\quad \text{for almost every}\hspace{2mm}t\in\mathbb{R},
\end{equation*}
with similar considerations for the conservation of angular momentum and kinetic energy. Finally, as 
\begin{equation*}
|x^{\varepsilon}(t)-\ov{x}^{\varepsilon}(t)|\geq 1\quad \text{for all}\hspace{2mm}t\leq 0\hspace{2mm}\text{and}\hspace{2mm}t\geq \tau_{+}^{\varepsilon},
\end{equation*}
it follows that $|x_{\infty}(t)-\ov{x}_{\infty}(t)|\geq 1$ for almost every $t\in\mathbb{R}$. As such, it follows that $Z_{\infty}:=[x_{\infty}, \ov{x}_{\infty}, v_{\infty}, \ov{v}_{\infty}]$ is a weak solution of system \eqref{sminus} and \eqref{splus}. \qed

We now show that the equivalence class $Z_{\infty}$ has a representative which is the unique classical solution of system \eqref{sminus} and \eqref{splus}. Indeed, we have the following proposition.
\begin{prop}
For $Z_{0}\in\Sigma^{-}$ and any $-\infty<T_{0}<T_{1}<\infty$, one has
\begin{equation*}
\lim_{\varepsilon\rightarrow 0}\int_{T_{0}}^{T_{1}}|v^{\varepsilon}(t)-v(t)|\,dt=\lim_{\varepsilon\rightarrow 0}\int_{T_{0}}^{T_{1}}|\ov{v}^{\varepsilon}(t)-\ov{v}(t)|\,dt=0,
\end{equation*}
where $v$ and $\ov{v}$ are the unique classical solutions \eqref{vclass} and \eqref{ovvclass} of system \eqref{sminus} and \eqref{splus}.
\end{prop}
\begin{proof}
It is necessary to break the proof into 3 cases, namely (i) $T_{1}<0$, (ii) $T_{0}\leq 0< T_{1}$, and (iii) $0<T_{0}$. We consider here only the most involved case, namely (ii). Indeed, supposing that $T_{0}\leq 0<T_{1}$, we break the $L^{1}(I, \mathbb{R}^{6})$ norm of $V^{\varepsilon}-V$ into three pieces:
\begin{equation*}
\displaystyle \|v^{\varepsilon}-v\|_{L^{1}(I, \mathbb{R}^{3})} =\displaystyle \underbrace{\int_{T_{0}}^{0}|v^{\varepsilon}(s)-v(s)|\,ds}_{J_{1}^{\varepsilon}}+\displaystyle \underbrace{\int_{0}^{2\tau_{\ast}^{\varepsilon}}|v^{\varepsilon}(s)-v(s)|\,ds}_{J_{2}^{\varepsilon}:=} +\displaystyle \underbrace{\int_{2\tau_{\ast}^{\varepsilon}}^{T_{1}}|v^{\varepsilon}(s)-v(s)|\,ds}_{J_{3}^{\varepsilon}:=}.
\end{equation*}
Clearly, $J_{1}^{\varepsilon}=0$. For $J_{2}^{\varepsilon}$, we find that
\begin{equation*}
\begin{array}{ll}
J_{2}^{\varepsilon} & =\displaystyle\int_{0}^{2\tau_{\ast}^{\varepsilon}}\left|\int_{0}^{t}\nabla\Phi^{\varepsilon}(x^{\varepsilon}(s)-\ov{x}^{\varepsilon}(s))\,ds+v_{n}'-v_{0}\right|\,dt\vspace{2mm}\\
& \leq \displaystyle\frac{1}{\varepsilon}\int_{0}^{2\tau_{\ast}^{\varepsilon}}\int_{0}^{t}|\Phi_{0}'(\rho^{\varepsilon}(s))|\,dsdt +2\tau_{\ast}^{\varepsilon}|v_{n}'-v_{0}|.
\end{array}
\end{equation*}
Since $\rho_{\ast}^{\varepsilon}\leq \rho^{\varepsilon}(s)\leq 1$, for $\varepsilon$ small enough, we have $\max_{0\leq s\leq 2\tau_{\ast}^{\varepsilon}}|\Phi_{0}'(\rho^{\varepsilon}(s))|=|\Phi_{0}'(\rho_{\ast}^{\varepsilon})|$, we infer that
\begin{equation*}
J_{2}^{\varepsilon}\leq \kappa_{1}q_{1}^{\beta-1}\varepsilon^{-1/\beta}\int_{0}^{2\tau_{\ast}^{\varepsilon}}dsdt+C\varepsilon^{1/\beta},
\end{equation*}
and so $J_{2}^{\varepsilon}\leq C\varepsilon^{1/\beta}$ for some constant $C>0$. For the final piece $J_{3}^{\varepsilon}$, we note that
\begin{equation*}
J_{3}^{\varepsilon}\leq (T_{1}-2\tau_{\ast}^{\varepsilon})\left|(\Pi_{2}\sigma^{\varepsilon}Z_{0})_{1}-v_{n}'\right|,
\end{equation*}
however corollary \ref{compactcon} immediately gives us that $\lim_{\varepsilon\rightarrow 0}|(\Pi_{2}\sigma^{\varepsilon}Z_{0})_{1}-v_{n}'|=0$. As the considerations for $\ov{v}^{\varepsilon}-\ov{v}$ are identical, the proof of the proposition follows.
\end{proof}
We note using proposition \ref{weakstar} that the {\em whole family} $\{V^{\varepsilon}\}_{0<\varepsilon<1}$ (and not simply a subsequence thereof) converges weakly-$\ast$ to $V$ in $\mathrm{BV}(I, \mathbb{R}^{6})$ for any open interval $I$. A standard application of the triangle inequality yields equality of $[v_{\infty}, \ov{v}_{\infty}]$ and (the equivalence class generated by) $[v, \ov{v}]$ in $L^{1}_{\mathrm{loc}}(\mathbb{R}, \mathbb{R}^{6})$, from which the proof of theorem \ref{mainresy} is concluded.

\section{Brief Remarks on the Non-Spherical Particle Problem}\label{close}
While the problem of two spherical particles is well understood, the problem for two {\em non-spherical} particles is less so. We now discuss some of the challenges one must face, both from the point of view of existence theory and regularity theory, for the analogous problem for the dynamics of non-spherical particles in $\mathbb{R}^{3}$. This brief section is by no means comprehensive, but it serves to highlight a few of the interesting open problems in the theory of hard particle dynamics.
\subsection{Set-up of the Problem}
Let $\mathsf{P}_{\ast}$ be a compact, strictly-convex, non-spherical subset of $\mathbb{R}^{3}$ with boundary surface of class $C^{1}$, whose centre of mass lies at the origin. We denote by $m>0$ and $J\in\mathbb{R}^{3\times 3}$ the mass and the inertia tensor associated to the reference particle $\mathsf{P}_{\ast}$, respectively, where
\begin{equation*}
m:=\int_{\mathsf{P}_{\ast}}dy\quad \text{and}\quad J:=\int_{\mathsf{P}_{\ast}}\left(y\otimes y-|y|^{2}I\right)\,dy.
\end{equation*}
We consider the motion in $\mathbb{R}^{3}$ of two congruent copies $\mathsf{P}$ and $\ov{\mathsf{P}}$ of the reference particle $\mathsf{P}_{\ast}$. Indeed, we may write $\mathsf{P}=R\mathsf{P}_{\ast}+x$ and $\ov{\mathsf{P}}=\ov{R}\mathsf{P}_{\ast}+\ov{x}$ for some maps $x, \ov{x}:\mathbb{R}\rightarrow\mathbb{R}^{3}$ and $R, \ov{R}:\mathbb{R}\rightarrow\mathrm{SO}(3)$. The phase space for this problem is
\begin{equation*}
\mathcal{D}_{2}(\mathsf{P}_{\ast}):=\left\{
Z\in (\mathbb{R}^{3}\times\mathrm{SO}(3)\times \mathbb{R}^{3}\times\mathbb{R}^{3})^{2}\,:\,\mathrm{card}\, (R\mathsf{P}_{\ast}+x)\cap(\ov{R}\mathsf{P}_{\ast}+\ov{x})\leq 1
\right\}.
\end{equation*}
The equations of motion for $\mathsf{P}$ and $\ov{\mathsf{P}}$ are given by
\begin{equation}\label{nonspher}
\frac{d}{dt_{\pm}}\left[
\begin{array}{c}
x \\
\ov{x}\\
R\\
\ov{R}
\end{array}
\right]=\left[
\begin{array}{c}
v \\
\ov{v}\\
\Omega R \\
\ov{\Omega}\ov{R}
\end{array}
\right]\qquad \text{and}\qquad \frac{d}{dt_{\pm}}\left[
\begin{array}{c}
v\\
\ov{v}\\
RJR^{T}\omega\\
\ov{R}J\ov{R}^{T}\ov{\omega}
\end{array}
\right]=\left[
\begin{array}{c}
0 \\
0\\
0\\
0
\end{array}
\right],\tag{$\mathsf{P}_{\ast}$S$_{\pm}$}
\end{equation}
where
\begin{equation*}
\Omega:=\left(
\begin{array}{ccc}
0 & -\omega_{3} & \omega_{2} \\
\omega_{3} & 0 & -\omega_{1}\\
-\omega_{2} & \omega_{1} & 0
\end{array}
\right)\quad\text{and}\quad\ov{\Omega}:=\left(
\begin{array}{ccc}
0 & -\ov{\omega}_{3} & \ov{\omega}_{2} \\
\ov{\omega}_{3} & 0 & -\ov{\omega}_{1}\\
-\ov{\omega}_{2} & \ov{\omega}_{1} & 0
\end{array}
\right),
\end{equation*}
and $\omega:=[\omega_{1}, \omega_{2}, \omega_{3}]$, $\ov{\omega}:=[\ov{\omega}_{1}, \ov{\omega}_{2}, \ov{\omega}_{3}]$. 
We now establish the analogues of classical and weak solution to system \eqref{nonspher}.
\begin{defn}[Classical Solutions]
For $Z_{0}\in\mathcal{D}_{2}(\mathsf{P}_{\ast})$, we say that $Z=[z, \ov{z}]$, with $z=[x, R, v, \omega]$ and $\ov{v}=[\ov{x}, \ov{R}, \ov{v}, \ov{\omega}]$, is a {\em classical solution} of system \eqref{nonspher} if and only if $t\mapsto Z(t)$ has range in $\mathcal{D}_{2}(\mathsf{P}_{\ast})$ for all $t\in\mathbb{R}$, $t\mapsto [x, \ov{x}, R, \ov{R}]$ is continuous piecewise linear on $\mathbb{R}$, and $t\mapsto[v, \ov{v}, \omega, \ov{\omega}]$ is lower semi-continuous and piecewise constant on $\mathbb{R}$; moreover, $t\mapsto Z(t)$ satisfies the right limit ODEs pointwise everywhere in $\mathbb{R}$ and the left limit ODEs pointwise everywhere in $\mathbb{R}\setminus\mathcal{T}(Z_{0})$, where 
\begin{equation*}
\mathcal{T}(Z_{0}):=\left\{t\in\mathbb{R}\,:\,\mathrm{card}\,(R(t)\mathsf{P}_{\ast}+x(t))\cap (\ov{R}(t)\mathsf{P}_{\ast}+\ov{x}(t))=1\right\}.
\end{equation*}
Additionally, $t\mapsto Z(t)$ must satisfy the conservation of total linear momentum
\begin{equation}\label{nlm}
mv(t)+m\ov{v}(t)=mv_{0}+m\ov{v}_{0},
\end{equation}
the conservation of total angular momentum (with respect to all points of measurement $a\in\mathbb{R}^{3}$)
\begin{equation*}
\begin{array}{c}\label{nam}
-m(a-x(t))\wedge v(t)+R(t)JR(t)^{T}\omega(t)-m(a-\ov{x}(t))\wedge \ov{v}(t)+\ov{R}(t)J\ov{R}(t)^{T}\ov{\omega}(t) \vspace{1mm}\\
=-m(a-x_{0})\wedge v_{0}+R_{0}JR_{0}^{T}\omega_{0}-m(a-\ov{x}_{0})\wedge \ov{v}_{0}+\ov{R}_{0}J\ov{R}_{0}^{T}\ov{\omega}_{0},
\end{array}
\end{equation*}
and the conservation of total kinetic energy
\begin{equation}\label{nke}
\begin{array}{c}
m|v(t)|^{2}+R(t)JR(t)^{T}\omega(t)\cdot \omega(t)+m|\ov{v}(t)|^{2}+\ov{R}(t)J\ov{R}(t)^{T}\ov{\omega}(t)\cdot \ov{\omega}(t)\vspace{1mm}\\
=m|v_{0}|^{2}+R_{0}JR_{0}^{T}\omega_{0}\cdot \omega_{0}+m|\ov{v}_{0}|^{2}+\ov{R}_{0}J\ov{R}_{0}^{T}\ov{\omega}_{0}\cdot \ov{\omega}_{0}
\end{array}
\end{equation}
for all $t\in\mathbb{R}$. Finally, $Z(0)=Z_{0}$.
\end{defn}
We contrast this with the following natural notion of weak solution.
\begin{defn}[Weak Solutions]
For $Z_{0}\in\mathcal{D}_{2}(\mathsf{P}_{\ast})$, we say that $Z=[z, \ov{z}]$ is a {\em weak solution} of system \eqref{nonspher} if and only if $[x, \ov{x}, R, \ov{R}]\in C(\mathbb{R}, \mathbb{R}^{6}\times\mathrm{SO}(3)^{2})$, $[v, \ov{v}, \omega, \ov{\omega}]\in \mathrm{BV}_{\mathrm{loc}}(\mathbb{R}, \mathbb{R}^{12})$ and $Z$ satisfies the equations
\begin{equation*}
\int_{-\infty}^{\infty}\left[
\begin{array}{c}
x \\
\ov{x} \\
R \\
\ov{R}
\end{array}
\right]\cdot \phi'\,dt=-\int_{-\infty}^{\infty}\left[ 
\begin{array}{c}
v \\
\ov{v}\\
\Omega R\\
\ov{\Omega}\ov{R}
\end{array}
\right]\cdot \phi\,dt
\end{equation*}
for all compactly-supported $C^{1}$ test functions $\phi$, and 
\begin{equation*}
\int_{-\infty}^{\infty}\left[
\begin{array}{c}
v \\
\ov{v}\\
\omega \\
\ov{\omega}
\end{array}
\right]\cdot\psi'\,dt=-\int_{-\infty}^{\infty}\psi\,dDV
\end{equation*}
for all compactly-supported $C^{1}$ test functions $\psi$, with $DV$ a finite vector-valued Radon measure on $\mathbb{R}$. Moreover, $Z$ should satisfy the conservation laws \eqref{nlm}, \eqref{nam} and \eqref{nke} pointwise almost everywhere on $\mathbb{R}$. Finally, any representative of the equivalence class $Z$ should have range in $\mathcal{D}_{2}(\mathsf{P}_{\ast})$ for almost every time in $\mathbb{R}$.
\end{defn}
We now catalogue some results in the literature pertaining to the existence and regularity of both classical and weak solutions to the equations of physical non-spherical particle motion.
\subsection{Existence Theories for System \eqref{nonspher}}
Let us firstly consider the problem of establishing the existence of classical solutions to the above system of ODEs by means of the method of trajectory surgery. In order to employ this method, one must first provide the analogue of the scattering matrices \eqref{boltzmannscattering} which resolve the collision between two non-spherical particles. We note that a collision between $\mathsf{P}$ and $\ov{\mathsf{P}}$ is parametrised by $\beta=(R, \ov{R}, n)\in\mathrm{SO}(3)\times\mathrm{SO}(3)\times\mathbb{S}^{2}$. We begin with the following.
\begin{defn}[Distance of Closest Approach]
For a given spatial configuration $\beta=(R, \ov{R}, n)\in\mathrm{SO}(3)\times\mathrm{SO}(3)\times\mathbb{S}^{2}$ of two colliding particles, we write $d_{\beta}$ to denote the {\em distance of closest approach} between the centres of mass of the two particles $\mathsf{P}$ (with centre of mass 0 and orientation $R$) and $\ov{\mathsf{P}}$ (with centre of mass on the line $\{dn\,:\,d>0\}$ and orientation $\ov{R}$), namely
\begin{equation*}
d_{\beta}:=\inf\left\{
d>0\,:\,\mathrm{card}\,(R\mathsf{P}_{\ast}\cap(\ov{R}\mathsf{P}_{\ast}+d n))=0
\right\}.
\end{equation*}
\end{defn}
We now state in precise way what we mean by a map which sends `pre-collisional' velocities to `post-collisional' velocities.
\begin{defn}[Scattering Maps]
For $\beta\in\mathrm{SO}(3)\times\mathrm{SO}(3)\times\mathbb{S}^{2}$, we say that $\sigma_{\beta}:\mathbb{R}^{12}\rightarrow\mathbb{R}^{12}$ is a {\em scattering map} if and only if $\sigma_{\beta}$ is an involution on $\mathbb{R}^{12}$ and maps the half space
\begin{equation*}
\Sigma_{\beta}^{-}:=\left\{V\in\mathbb{R}^{12}\,:\,\nu_{\beta}\cdot V\geq 0\right\}
\end{equation*}
to the half space
\begin{equation*}
\Sigma_{\beta}^{+}:=\left\{V\in\mathbb{R}^{12}\,:\,\nu_{\beta}\cdot V\leq 0\right\},
\end{equation*}
where $\nu_{\beta}\in\mathbb{R}^{12}$ is the unit vector satisfying the formal expressions
\begin{equation*}
\frac{d}{dt_{+}}F(x(t), \ov{x}(t), R(t), \ov{R}(t))\bigg|_{t=\tau}\leq 0\quad  \Longleftrightarrow\quad \nu_{\beta}\cdot V\geq 0
\end{equation*}
and 
\begin{equation*}
\frac{d}{dt_{-}}F(x(t), \ov{x}(t), R(t), \ov{R}(t))\bigg|_{t=\tau}\geq 0\quad  \Longleftrightarrow\quad \nu_{\beta}\cdot V\leq 0,
\end{equation*}
with $F:\mathbb{R}^{6}\times \mathrm{SO}(3)^{2}\rightarrow \mathbb{R}$ the auxiliary function
\begin{equation*}
F(x, \ov{x}, R, \ov{R}):=|x-\ov{x}|^{2}-d_{\beta(x, \ov{x})}^{2}, \qquad \beta(x, \ov{x}):=(R, \ov{R}, n(x, \ov{x}))\hspace{2mm}\text{and}\hspace{2mm} n(x, \ov{x}):=\frac{\ov{x}-x}{|\ov{x}-x|},
\end{equation*}
and where $t=\tau$ is a collision time, i.e. $F(x(\tau), \ov{x}(\tau), R(\tau), \ov{R}(\tau))=0$.
\end{defn}
With these basic concepts in place, and with the method of trajectory surgery in mind, let us state the problem of interest. 
\begin{op}[Characterisation of Physical $C^{1}$ Scattering Maps]\label{op1}
For every $\beta\in\mathrm{SO}(3)\times\mathrm{SO}(3)\times\mathbb{S}^{2}$, characterise all $C^{1}$ scattering maps $\sigma_{\beta}:\mathbb{R}^{12}\rightarrow\mathbb{R}^{12}$ which satisfy the Jacobian PDE
\begin{equation*}
\mathrm{det}(D\sigma_{\beta}[V])=-1 \quad \text{for}\hspace{2mm}V\in\mathbb{R}^{12},
\end{equation*}
and are subject to the algebraic constraints of conservation of linear momentum
\begin{equation*}
\left(\begin{array}{c}
(\sigma_{\beta}[V])_{1}\\
(\sigma_{\beta}[V])_{2}\\
(\sigma_{\beta}[V])_{3}\\
\end{array}\right)+\left(
\begin{array}{c}
(\sigma_{\beta}[V])_{4}\\
(\sigma_{\beta}[V])_{5}\\
(\sigma_{\beta}[V])_{6}
\end{array}
\right)=\left(
\begin{array}{c}
V_{1}\\
V_{2}\\
V_{3}
\end{array}
\right)+\left(
\begin{array}{c}
V_{4}\\
V_{5}\\
V_{6}
\end{array}
\right),
\end{equation*}
conservation of angular momentum (for any $a\in\mathbb{R}^{3}$)
\begin{equation*}
\begin{array}{c}
-ma\wedge\left(
\begin{array}{c}
(\sigma_{\beta}[V])_{1}\\
(\sigma_{\beta}[V])_{2}\\
(\sigma_{\beta}[V])_{3}
\end{array}
\right)+R^{T}JR\left(
\begin{array}{c}
(\sigma_{\beta}[V])_{7}\\
(\sigma_{\beta}[V])_{8}\\
(\sigma_{\beta}[V])_{9}
\end{array}
\right)-m(a-d_{\beta}n)\wedge\left(
\begin{array}{c}
(\sigma_{\beta}[V])_{4}\\
(\sigma_{\beta}[V])_{5}\\
(\sigma_{\beta}[V])_{6}
\end{array}
\right)+\ov{R}^{T}J\ov{R}\left(
\begin{array}{c}
(\sigma_{\beta}[V])_{10}\\
(\sigma_{\beta}[V])_{11}\\
(\sigma_{\beta}[V])_{12}
\end{array}
\right) \vspace{1mm}\\
=-ma\wedge\left(
\begin{array}{c}
V_{1} \\
V_{2}\\
V_{3}
\end{array}
\right)+R^{T}JR\left(
\begin{array}{c}
V_{7}\\
V_{8}\\
V_{9}
\end{array}
\right)-m(a-d_{\beta}n)\wedge\left(
\begin{array}{c}
V_{4}\\
V_{5}\\
V_{6}
\end{array}
\right)+\ov{R}^{T}J\ov{R}\left(
\begin{array}{c}
V_{10}\\
V_{11}\\
V_{12}
\end{array}
\right) 
\end{array}
\end{equation*}
and the conservation of kinetic energy
\begin{equation*}
|M\sigma_{\beta}[V]|^{2}=|MV|^{2},
\end{equation*}
where $M\in\mathbb{R}^{12\times 12}$ is the block mass-inertia matrix
\begin{equation*}
M:=\left(
\begin{array}{cccc}
\sqrt{m}I & 0 & 0 & 0 \\
0 & \sqrt{m}I & 0 & 0 \\
0 & 0 & \sqrt{J} & 0 \\
0 & 0 & 0 & \sqrt{J}
\end{array}
\right).
\end{equation*}
\end{op}
It has already been shown essentially in \textsc{Saint-Raymond and Wilkinson} \cite{2015arXiv150707601S} that the quasi-reflection matrices $\sigma_{\beta}\in\mathrm{O}(12)$ given by
\begin{equation}\label{alloa}
\sigma_{\beta}:=M^{-1}\left(I-2\widehat{\nu}_{\beta}\otimes \widehat{\nu}_{\beta}\right)M \quad \text{with}\hspace{2mm}\widehat{\nu}_{\beta}:=\frac{M^{-1}\nu_{\beta}}{|M^{-1}\nu_{\beta}|},
\end{equation}
are indeed physical $C^{1}$ scattering maps for any choice of $\beta$. However, according to the author's knowledge, it is not known if \eqref{alloa} is in any sense the unique solution of problem \ref{op1}. In any case, with at least one family of physical scattering matrices $\{\sigma_{\beta}\}_{\beta\in\mathrm{SO}(3)\times\mathrm{SO}(3)\times\mathbb{S}^{2}}$ in hand, one may make use of the results of \textsc{Ballard} \cite{MR1785473} to establish the following result on the global existence and regularity of weak solutions to system \eqref{nonspher}, {\em under the assumption that $\partial\mathsf{P}_{\ast}$ is real analytic}.
\begin{thm}[\textsc{Ballard} \cite{MR1785473}]
Suppose $\mathsf{P}_{\ast}$ is a compact, strictly-convex, non-spherical subset of $\mathbb{R}^{3}$ whose boundary surface is real analytic. Let $\{\sigma_{\beta}\}_{\beta\in \mathrm{SO}(3)\times\mathrm{SO}(3)\times\mathbb{S}^{2}}$ be a family of physical scattering maps. For any $Z_{0}\in\mathcal{D}_{2}(\mathsf{P}_{\ast})$, there exists a unique global-in-time weak solution to system \eqref{nonspher}.
\end{thm}
It seems somewhat unreasonable to stipulate that the regularity of the boundary $\partial\mathsf{P}_{\ast}$ be real analytic. We therefore draw attention to the following:
\begin{op}
Establish the global existence and regularity of solutions to \eqref{nonspher} when $\partial\mathsf{P}_{\ast}$ is only of class $C^{1}$.
\end{op}
If the dynamics of lower regularity particles does not exhibit any pathological behaviour (such as infinitely-many collisions in a compact time interval), then the topological methods outlined in this article provide, in principle, a method with which one could establish existence of classical solutions to system \eqref{nonspher}. We hope to address this problem in future work.

\bibliography{refo}

\end{document}